\documentclass[sigconf, nonacm, pdfa]{acmart}

\newcommand\vldbdoi{XX.XX/XXX.XX}
\newcommand\vldbpages{XXX-XXX}
\newcommand\vldbvolume{17}
\newcommand\vldbissue{11}
\newcommand\vldbyear{2024}
\newcommand\vldbauthors{\authors}
\newcommand\vldbtitle{\shorttitle} 
\newcommand\vldbavailabilityurl{https://github.com/LeixiaWang/PriPLT}
\newcommand\vldbpagestyle{plain}

\usepackage{dsfont} 
\usepackage[ruled,linesnumbered,vlined]{algorithm2e} 
\usepackage{enumitem} 
\setlist[enumerate]{wide=\parindent}
\usepackage[caption=false,font=scriptsize]{subfig} 
\usepackage{makecell} 
\usepackage{tikz} 
\usepackage{soul} 
\usepackage{diagbox} 
\usepackage{multirow} 
\usepackage[export]{adjustbox} 
\usepackage{colortbl} 

\usepackage{xmpincl}
\usepackage[a-2b]{pdfx}
\usepackage[T1]{fontenc}
\usepackage{aecompl}

\DeclareMathOperator{\var}{Var}
\DeclareMathOperator{\E}{\mathbb{E}}
\DeclareMathOperator{\cov}{Cov}

\allowdisplaybreaks[4]

\newcommand*\circled[1]{\tikz[baseline=(char.base)]{
    \node[shape=circle,draw,inner sep=0.5pt] (char) {\small #1};}}
    
\definecolor{c1}{HTML}{848484}
\definecolor{c2}{HTML}{3d3d3d}
\definecolor{c3}{HTML}{000000}
\definecolor{cb}{HTML}{E9EDEA}

\makeatletter
\newtheorem*{rep@theorem}{\rep@title}
\newcommand{\newreptheorem}[2]{%
\newenvironment{rep#1}[1]{%
 \def\rep@title{#2 \ref{##1}}%
 \begin{rep@theorem}}%
 {\end{rep@theorem}}}
\makeatother

\newreptheorem{theorem}{Theorem}

\begin{document}
\title{PriPL-Tree: Accurate Range Query for Arbitrary Distribution under Local Differential Privacy}

\author{Leixia Wang}
\orcid{0000-0003-3475-9552} 
\affiliation{%
  \institution{Renmin University of China}
}
\email{leixiawang@ruc.edu.cn}

\author{Qingqing Ye}
\orcid{0000-0003-1547-2847}
\affiliation{%
  \institution{Hong Kong Polytechnic University}
}
\email{qqing.ye@polyu.edu.hk}

\author{Haibo Hu}
\orcid{0000-0002-9008-2112}
\affiliation{%
  \institution{Hong Kong Polytechnic University}
}
\email{haibo.hu@polyu.edu.hk}

\author{Xiaofeng Meng}
\authornote{Corresponding author: Xiaofeng Meng.}
\orcid{0000-0002-7889-2120}
\affiliation{%
  \institution{Renmin University of China}
}
\email{xfmeng@ruc.edu.cn}

\begin{abstract}

Answering range queries in the context of Local Differential Privacy (LDP) is a widely studied problem in Online Analytical Processing (OLAP). Existing LDP solutions all assume a uniform data distribution within each domain partition, which may not align with real-world scenarios where data distribution is varied, resulting in inaccurate estimates. To address this problem, we introduce PriPL-Tree, a novel data structure that combines hierarchical tree structures with piecewise linear (PL) functions to answer range queries for arbitrary distributions. PriPL-Tree precisely models the underlying data distribution with a few line segments, leading to more accurate results for range queries. Furthermore, we extend it to multi-dimensional cases with novel data-aware adaptive grids. These grids leverage the insights from marginal distributions obtained through PriPL-Trees to partition the grids adaptively, adapting the density of underlying distributions. Our extensive experiments on both real and synthetic datasets demonstrate the effectiveness and superiority of PriPL-Tree over state-of-the-art solutions in answering range queries across arbitrary data distributions.

\end{abstract}

\maketitle

\pagestyle{\vldbpagestyle}
\begingroup\small\noindent\raggedright\textbf{PVLDB Reference Format:}\\
\vldbauthors. \vldbtitle. PVLDB, \vldbvolume(\vldbissue): \vldbpages, \vldbyear.\\
\href{https://doi.org/\vldbdoi}{doi:\vldbdoi}
\endgroup
\begingroup
\renewcommand\thefootnote{}\footnote{\noindent
This work is licensed under the Creative Commons BY-NC-ND 4.0 International License. Visit \url{https://creativecommons.org/licenses/by-nc-nd/4.0/} to view a copy of this license. For any use beyond those covered by this license, obtain permission by emailing \href{mailto:info@vldb.org}{info@vldb.org}. Copyright is held by the owner/author(s). Publication rights licensed to the VLDB Endowment. \\
\raggedright Proceedings of the VLDB Endowment, Vol. \vldbvolume, No. \vldbissue\ %
ISSN 2150-8097. \\
\href{https://doi.org/\vldbdoi}{doi:\vldbdoi} \\
}\addtocounter{footnote}{-1}\endgroup

\ifdefempty{\vldbavailabilityurl}{}{
\vspace{.3cm}
\begingroup\small\noindent\raggedright\textbf{PVLDB Artifact Availability:}\\
The source code, data, and/or other artifacts have been made available at \url{\vldbavailabilityurl}.
\endgroup
}

\section{Introduction}

With increasing personal information collected by third-party entities (a.k.a., data collectors), individual privacy protection is garnering more attention \cite{dwork2006calibrating, wang2023eps, zhang2023trajectory, ye2023stateful}. Local Differential Privacy (LDP) has emerged as a rigorous privacy-preserving standard widely employed in academia and industry \cite{erlingsson2014rappor, apple2017learning, ding2017collecting}. Under LDP, users only need to submit perturbed values, preserving the privacy of their raw data. The data collector collects these noisy values and invests effort in estimating various statistics to support data analysis tasks. 

Range queries, as a prevalent query type, have been extensively studied in LDP, where the data collector estimates the frequency of specific ranges within a domain. To support these queries, existing solutions construct hierarchical trees \cite{cormode2019answering, wang2019answering, du2021ahead, wang2023privnud} or grids \cite{yang2020answering, wang2022prism} over the whole domain and estimate frequencies of the partitioned subdomains (i.e., nodes in trees or cells in grids). To answer a range query, the frequencies of those nodes or cells covered by the given range will be summed up. When some subdomains are partially covered, the data within them is assumed to be uniformly distributed, so that the corresponding frequency can be estimated based on the overlap proportion with the query range. However, the data we indexed typically exhibits various distributions rather than uniform in reality. It is inevitable to introduce non-uniform errors by existing methods, leading to inaccurate responses.

Figure~\ref{fig:PL_example} shows an example of this non-uniform estimation error. Given a distribution depicted as the black curve, Figure~\ref{fig:PL_example}(a) partitions the domain in a coarse-grained manner, resulting in a large non-uniform error. Figure~\ref{fig:PL_example}(b) uses finer partitions to reduce the non-uniform error, but incurs significant aggregated LDP noise error due to an increasing number of bins.

\begin{figure}[h]
	\centering
	\includegraphics[width=0.48\textwidth]{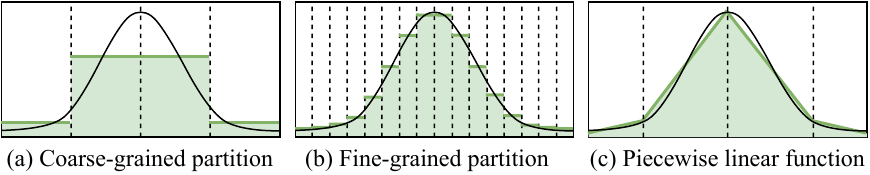}
	\caption{An Illustration on Non-uniform Errors}
	\label{fig:PL_example}
\end{figure}

To tackle this challenge, we propose an innovative solution that employs a \textbf{piecewise linear (PL) function} to model the underlying data distribution instead of relying on a uniform assumption. By partitioning the data domain into several intervals and approximating the data distribution within each interval with a line segment, even complex data distributions can be well-approximated with a few parameters \cite{Cosma2015optimal}. As shown in Figure~\ref{fig:PL_example}(c), the PL function accurately approximates the data distribution with a few segments (represented as frequency-slope pairs), which alleviates both non-uniform error and LDP noise error significantly. 

Building upon the PL function, we introduce the \textbf{Private Piecewise Linear Tree (PriPL-Tree)}. In this tree, each leaf node corresponds to a line segment and stores a frequency-slope pair, while each non-leaf node represents an interval combined from its child nodes' intervals and stores the associated interval frequency. Compared to traditional hierarchical trees, the PriPL-Tree offers several significant advantages: (1) Within a node with the same interval, it provides a more accurate fit to the underlying data distribution than the uniform assumption. (2) A few segments are sufficient to model the distribution, resulting in fewer leaf nodes and a lower tree height. Height reduction is crucial in LDP as it facilitates allocating users or the privacy budget among fewer tree layers (a necessary step to meet LDP's privacy guarantee), thereby mitigating noise errors in frequency estimation for each node. (3) The number of parameters in the tree depends only on the shape of the data distribution, not the domain size, enabling adaptation to large domains with a more concise structure and more accurate results.

However, constructing the PriPL-Tree under LDP is non-trivial because the data remains invisible to the data collector. To address this, we propose a three-phase approach. First, we allocate a portion of users to estimate the data histogram, gaining a rough glimpse of the underlying data distribution, and use it to fit the PL function. Next, we construct the optimal PriPL-Tree and estimate node frequencies with the remaining users. Finally, we perform post-processing to refine the frequencies and slopes in the tree, ensuring the non-negativity and consistency of nodes in the tree.

In addition to handling 1-D range queries, we extend PriPL-Tree for multi-dimensional scenarios by incorporating 2-D adaptive grids. These adaptive grids are also data-aware, featuring non-uniform partitions that adapt to the density of the data distribution, and can be constructed utilizing marginal distributions from 1-D PriPL-Trees. By leveraging both 1-D PriPL-Trees and 2-D grids, we can answer $\lambda$-D range queries ($\lambda\!>\!1$) using the weighted updating approach \cite{yang2020answering, wang2022prism, wang2023privnud}. 

To summarize, our contributions are:
\begin{itemize}[leftmargin=0.5cm, itemindent=0.0cm]
	\item \textbf{Innovative PriPL-Tree}: We design PriPL-Tree, a novel data structure that models the underlying data distribution using a piecewise linear (PL) function instead of relying on a uniform data assumption in LDP. In this way, PriPL-Tree can answer range queries for arbitrary data distributions accurately.  
	\item \textbf{Adaptive data-aware Grids}: By leveraging marginal distributions revealed by the PriPL-Trees, we design adaptive grids tailored to the density of underlying data distributions, which serves as the building block for answering multi-dimensional range queries effectively.
	\item \textbf{Extensive Experimental Evaluation}: We conduct comprehensive experiments on both real and synthetic datasets, validating the effectiveness and superiority of our methods. Compared to existing approaches, our method achieves one order of magnitude improvement in accuracy. 
\end{itemize}

In the remainder of this paper, we introduce LDP and analyze existing range query methods in Section~\ref{sec:preliminary}. Our primary method, PriPL-Tree, is proposed in Section~\ref{sec:pripl-tree}, extended with adaptive grids for multi-dimensional cases in Section~\ref{sec:adaptive_grid}. We evaluate them in Section~\ref{sec:experiments}. Finally, we review related works in Section~\ref{sec:related_work} and conclude our paper in Section~\ref{sec:conclusion}.

\section{Preliminaries}
\label{sec:preliminary}

In this section, we define the problem and introduce necessary knowledge of LDP and existing methods for range queries in LDP.

\subsection{Local Differential Privacy (LDP)}
\label{subsec:ldp_mechanism}
In the context of data collection, LDP provides a mechanism $\mathcal{R}$ that enables users to perturb their data $v$ before sharing it with an untrusted data collector \cite{duchi2018minimax, qian2023collaborative, duan2024ldptube}. By ensuring the resulting perturbed data $\mathcal{R}(v)$ satisfies $\epsilon$-LDP, the data collector cannot distinguish a value $v$ from any other possible value $v'$ with high confidence, thus safeguarding the privacy. A higher level of privacy is achieved when a smaller value of $\epsilon$ is employed.

\begin{definition}[$\epsilon$-Local Differential Privacy ($\epsilon$-LDP) \cite{duchi2018minimax}]
A perturbation mechanism $\mathcal{R}$ satisfies $\epsilon$-LDP ($\epsilon>0$) iff for any pair of input data $v, v' \in D$ and any output $z$ of $\mathcal{R}$, we have
\begin{equation*}
	\textstyle
	\Pr[\mathcal{R}(v)=z] \le e^{\epsilon} \Pr[\mathcal{R}(v')=z].
\end{equation*}
\end{definition}

We introduce two state-of-the-art LDP mechanisms for fundamental frequency and numerical distribution estimation, respectively, both ensuring $\epsilon$-LDP.

\textbf{Optimal Unary Encoding Mechanism (OUE)} \cite{wang2017locally} is the state-of-the-art frequency estimation mechanism with three steps: encoding, perturbation, and aggregation. In the encoding step, each user $u_i$ ($i\le N$) encodes his value $v_i\!\in\!D$ into a bit vector $\mathbf{B}\!\in\!\{0,1\}^{|D|}$, setting the $i$-th position to 1 and others to 0. During perturbation, each user perturbs each bit in $\mathbf{B}$ separately. The original bit ``1'' is retained with probability $p\!=\!1/2$, while the bit ``0'' is flipped to ``1'' with probability $q\!=\!1/(e^\epsilon\!+\!1)$. Then, the data collector aggregates all $N$ users' perturbed vectors, counts the number of $1$s in the $v$-th position as $n_v^\prime$ for each $v$, and calibrates it to an unbiased frequency estimate $\hat{f}_v\!=\!(n_v^\prime\!-\!Nq) / N(p\!-\!q)$, achieving an optimized estimation variance of $\var(\hat{f}_v) \approx 4e^{\epsilon} / (N\cdot(e^\epsilon\!-\!1)^2)$, denoted as $\sigma^2$.

\textbf{Square Wave Mechanism (SW)} \cite{li2020estimating} is for numerical distribution estimation, involving perturbation and aggregation steps. In the perturbation step, each user perturbs his value $v_i \in D$ to $v^\prime_i$ within a domain with size $|D| + 2b$, where $b\!=\!\left\lfloor\frac{\epsilon e^\epsilon - e^\epsilon + 1}{2e^\epsilon(e^\epsilon-1-\epsilon)}\cdot|D|\right\rfloor$. Specifically, with a larger probability $p\!=\!e^\epsilon/((2b\!+\!1)e^\epsilon\!+\!|D|\!-\!1)$, he perturbs $v_i$ to a value $v^\prime_i$ within $|v-v^\prime_i|<b$; with a smaller probability $q\!=\!1/((2b\!+\!1)e^\epsilon\!+\!|D|\!-\!1)$, he perturbs it to other values. During aggregation, the data collector collects the perturbed data and estimates the distribution using the expectation maximization (EM) algorithm or the EM algorithm with smoothing steps (EMS). SW with EM captures spiky distributions effectively, while SW with EMS provides more accurate estimation by smoothing the LDP noise. We denote these two results as $\hat{\mathbf{F}}^{\text{EM}}$ and $\hat{\mathbf{F}}^{\text{EMS}}$, respectively.


\subsection{Problem Definition}

Consider $N$ users and each user $u_i$ ($i \le N$) owns a private record containing $m$ private values on attributes $(A_1, A_2, \ldots, A_m)$. Each attribute $A_j$ ($1\le j\le m$) has a public domain $D_j$. Each user $u_i$'s record is denoted as $\mathbf{v}_i = (v_i^1, v_i^2, \ldots, v_i^m)$, where $v_i^j \in D_j$ represents the attribute $A_j$'s value for user $u_i$. For convenience, we assume $D_j = [0, d_j]$ for continuous data and $D_j = \{1, 2, \ldots, d_j\}$ (abbreviated as $[d_j]$) for discrete data. For 1-dimensional (a.k.a., 1-D) data, we abuse $v_i$ to denote the user $u_i$'s value in the default attribute.

The $\lambda$-dimensional (a.k.a., $\lambda$-D) range query is performed on a set of private attributes $\Phi\!\subseteq\!\{A_j|j\le m\}$, where $\lambda\!=\!|\Phi|\!\le\!m$. Let $[l_j, r_j]$ denote the specified range for the attribute $A_j \in \Phi$. The $\lambda$-D range query returns the frequency of records where all queried attribute values $v_i^j$ ($j\in\Phi$) are within these specified ranges. Formally, 
\begin{equation*}
\textstyle
	Q\left(\cap_{A_j\in\Phi}[l_j,r_j]\right) = \frac{1}{N}\sum_{i=1}^{N}\mathds{1}_{\bigcap_{A_j\in\Phi}\{l_j \le v_i^j \le r_j\}},
\end{equation*}
where $\mathds{1}$ is an indicator function that outputs $1$ if the predicate is true and $0$ otherwise. 

Our goal is to let the untrusted data collector answer the range query $Q\left(\cap_{j\in\Phi}[l_j,r_j]\right)$ while ensuring individual privacy under $\epsilon$-LDP. Extensive research has been conducted on this problem in the context of LDP, as we reviewed below. The notations used are summarized in Table~\ref{tab:notations}. 

\begin{table}
	\small
	\setlength\tabcolsep{2pt} 
	\caption{Notations}
	\begin{center}
		\begin{tabular}{|c|c|}
			\hline
			\textbf{Symbols} & \textbf{Description} \\
			\hline
			$N$ & The total number of users \\
			\hline
			$A_j$ & The $j$-th attribute\\
			\hline
			$D_j, d_j$ & The attribute $A_j$'s domain $D_j$ with size $d_j$ \\
			\hline
			$m$ & The number of private attributes in the data \\
			\hline
			$\lambda$ & The number of attributes involved in a range query \\
			\hline
			$n_k$ & The node $n_k$ in the PriPL-Tree\\
			\hline
			$f_k$, $\beta_k$ & The frequency $f_k$ and the slope $\beta_k$ in $n_k$\\
			\hline
			$I_k, s_{k-1},s_k$ & \makecell[l]{The interval $I_k$ of node $n_k$ including $|I_k|$ bucketized values\\ between two breakpoints $s_{k-1}$ and $s_k$}\\
			\hline
			$\alpha$ & \makecell[c]{User allocation ratio in phase 1 for PriPL-Tree}\\
			\hline
			$\sigma^2$ & The variance of OUE with $N$ users and a privacy budget of $\epsilon$\\
			\hline
		\end{tabular}
		\label{tab:notations}
	\end{center}
\end{table}

\subsection{Existing Methods}

The hierarchical tree (HT) is the primary data structure for 1-D range queries in LDP. It hierarchically decomposes the entire domain into disjoint sub-domains (a.k.a., intervals), constructing a $B$-ary tree. Each node in the tree represents an interval and stores an estimated interval frequency. Non-leaf nodes aggregate frequencies of their $B$ child nodes. The data within leaf nodes is assumed to be uniformly distributed. As such, range queries can be answered by summing a few node frequencies (or parts of them) in the tree rather than all individual bins' frequencies within the range, as in a histogram, reducing the accumulated noise error. 
For example, considering a domain $D\!=\![0,16]$, we can either uniformly partition it into 16 bins for a histogram or construct a complete binary tree with 16 leaves. Given a range query $Q([0,5])$, it can be answered by summing the two frequencies of nodes with intervals $[0,4)$ and $[4,5]$ in the tree, rather than five frequencies of individual bins $[0,1)$, $[1,2)$, $[2,3)$, $[3,4)$ and $[4,5]$ in the histogram.
To build this tree under LDP, the privacy budget or users are allocated among layers to estimate nodes' frequencies. To further reduce the error of range queries, a lot of optimization methods have been developed, including Haar transformation of data~\cite{cormode2019answering}, optimizing the branch number for the tree \cite{cormode2019answering,wang2019answering}, customizing branch numbers for nodes \cite{wang2023privnud} and merging nodes with low frequencies~\cite{du2021ahead}. 

Beyond 1-D, the HT can be extended to multi-dimensional cases \cite{wang2019answering, du2021ahead}. However, finely partitioning users or the privacy budget among layers and dimension combinations would increase the noise error. To overcome the curse of dimensionality, grid-based methods are typically employed in $\lambda$-D ($\lambda\!>\!1$) range queries. Given $m$ private attributes, they estimate the frequencies of $m$ 1-D grids for individual attributes and $\binom{m}{2}$ 2-D grids for attribute pairs. Users or privacy budgets are allocated among these grids, where the data in each cell is still assumed to be uniformly distributed. Based on these, a $\lambda$-D range query can be estimated from these 1-D and 2-D grids through the maximum entropy \cite{zhang2018calm} or weighted updating \cite{yang2020answering} algorithms. To reduce the estimation error of range queries, Yang et al. \cite{yang2020answering} optimized the granularity for both 1-D and 2-D grids, and Wang et al. \cite{wang2022prism} further employed prefix-sum (PS) cubes.

\subsection{Observations and Challenges} 
\label{subsec:summary}
Drawing from current research on range queries in LDP, we summarize two key observations that guide our approach and identify a significant challenge. First, we outline the observations:

\textbf{(1) Tree vs. Grid}: Tree-based methods allocate users (or privacy budget) across multiple layers and dimensions, whereas grid-based methods allocate only among dimensions. Considering the significant noise from a few users or a small privacy budget, tree-based methods are preferable for 1-D range queries, while grid-based methods are more suitable for $\lambda$-D ($\lambda>1$) range queries \cite{wang2023privnud}.
	
\textbf{(2) User Allocation vs. Privacy Budget Allocation}: To achieve $\epsilon$-LDP, allocation of users or privacy budget is necessary among the layers of trees and the grids. Generally, user allocation is preferable in the LDP setting as it introduces less noise error than privacy budget allocation \cite{cormode2019answering, wang2019answering, yang2020answering, du2021ahead, wang2023privnud}.	

We then present a significant challenge: \textbf{unrealistic uniform assumptions}. All existing works decompose the domain uniformly and/or assume uniform data distribution in each decomposed sub-domain \cite{cormode2019answering, wang2019answering, yang2020answering, du2021ahead, wang2022prism, wang2023privnud}. However, real-world applications often involve data following various distributions (e.g., Gaussian, Zipf) rather than being uniformly distributed \cite{li2014data, wei2014partial}. This assumption inevitably leads to non-uniform errors and suboptimal estimates. 

In this work, we address uniform assumptions on the domain decomposition and data, and correspondingly provide enhanced estimation accuracy for range queries in LDP. In what follows, we first present a solution for 1-D range queries in Section~\ref{sec:pripl-tree} and then extend it to a multi-dimensional setting in Section~\ref{sec:adaptive_grid}. 

\section{Private Piecewise Linear Tree}
\label{sec:pripl-tree}


In this section, we propose PriPL-Tree, a private piecewise linear tree that combines piecewise linear functions and hierarchical trees to address uniform assumptions for 1-D range queries.

\subsection{Design Rationale}

The piecewise linear (PL) function is capable of approximating the underlying data distribution with only a few parameters, enabling us to not rely on uniform distribution assumptions.
For instance, given a Gaussian distribution in Figure~\ref{fig:priplt_example}(a), we can approximate it using 4 segments with 8 parameters. In this case, the entire domain is divided into 4 intervals, and data in each interval $I_k\!=\![s_{k-1},s_k)$ ($1\le k\le 4$) is fitted with a linear function defined by two parameters, i.e., $\beta_k$ (slope) and $f_k$ (sum of frequencies of all points in the interval). The linear expression is given by $y = \beta_k x + b_k$, where $b_k = f_k / |I_k| - \beta_k (|I_k|+ 2s_{k-1}-1)/2$ and $|I_k|$ is the interval size. 

To facilitate range query processing, we integrate the PL function with a hierarchical tree structure, proposing the \textbf{Private Piecewise Linear Tree (PriPL-Tree)}. Each leaf node represents a segment (corresponding to an interval) of the PL function and stores its slope $\beta_k$ and frequency $f_k$. Each non-leaf node represents an interval and only stores the interval frequency. Like conventional hierarchical trees, the parent node stores the sum of its child nodes' frequencies. We count the layers of the tree starting from 0 at the top.

\begin{figure}[h]
\centering
\includegraphics[width=0.47\textwidth]{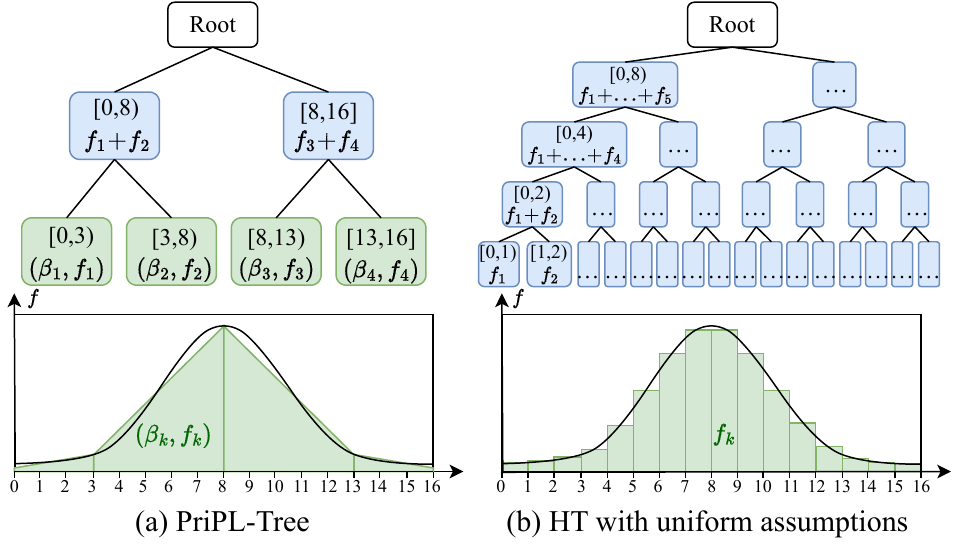}
\caption{An Example of PriPL-Tree and HT}
\label{fig:priplt_example}
\end{figure}

We provide an example in Figure~\ref{fig:priplt_example} to illustrate the PriPL-Tree, comparing it to the conventional hierarchical tree (HT) with uniform assumptions. 
Obviously, within the same interval, the PL function can provide a more accurate approximation of the underlying distribution than uniform assumptions. As such, the PriPL-Tree captures the underlying distribution using significantly fewer leaf nodes (4 in PriPL-Tree vs. 16 in HT) and correspondingly fewer layers (2 in PriPL-Tree vs. 4 in HT). In the context of LDP, fewer layers mean each layer in the tree can be allocated more users, resulting in less noise error due to the law of large numbers. Moreover, the PriPL-Tree construction depends solely on the distribution of the underlying data, as opposed to HT, which relies on the domain size. In HT, modeling data with a large domain size requires a taller tree or a coarser granularity for leaf nodes, increasing noise errors or non-uniform errors. The PriPL-Tree is well-suited for large domain-sized scenarios while reducing both two types of errors.

However, constructing an effective PriPL-Tree in LDP settings is challenging due to the invisible data distribution. To address this, we first employ some users to collect a noisy histogram using LDP mechanisms, gaining insight into the underlying data distribution. We then fit PL functions based on this noisy histogram and use the remaining users to construct the tree. Through post-processing, we further optimize these estimated frequencies and slopes to maintain tree consistency and improve range query accuracy. Following this idea, we propose a three-phase workflow as outlined below and detail the methods for each phase in separate subsections.

\subsection{Workflow of PriPL-Tree}
\label{subsec:workflow}

The workflow of the PriPL-Tree involves three phases: Private PL Fitting, PriPL-Tree Construction, and PriPL-Tree Refinement, exemplified in Figure~\ref{fig:framework}.

\begin{figure*}[htb]
\centering
\includegraphics[width=\textwidth]{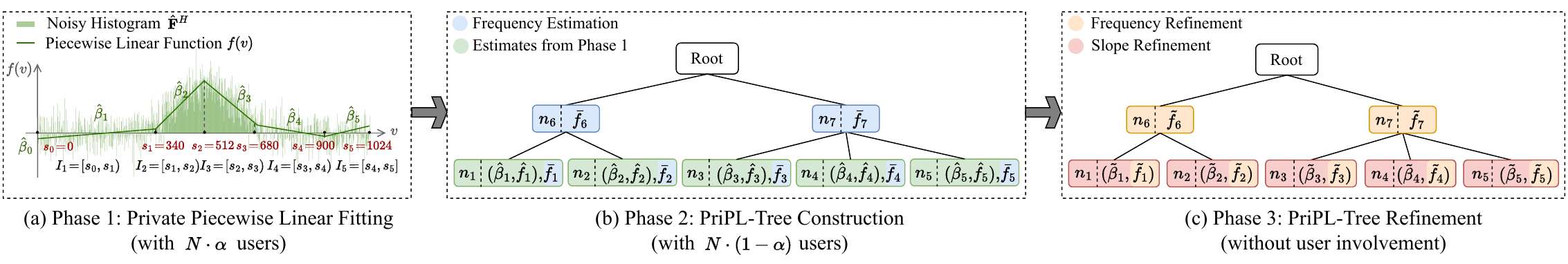}
\caption{Workflow of PriPL-Tree}
\label{fig:framework}
\end{figure*}

\textbf{Phase 1: Private Piecewise Linear (PL) Fitting.}
To gain fundamental insight into the data distribution, we employ a proportion $\alpha$ of users to execute SW protocols with the privacy budget $\epsilon$, collecting a noisy histogram $\hat{\mathbf{F}}^{\text{H}}$ on the bucketized domain $[d]$. Then, we fit the PL function over this histogram, as presented in Section~\ref{subsec:pl_fitting}. During PL fitting, we address two key issues: interval partitioning and segment fitting. Interval partitioning involves determining the number of intervals $K$ and identifying $K\!+\!1$ breakpoints $\{s_0, s_1, ..., s_K\}$. The derived intervals are denoted as $I_k=[s_{k-1}, s_k)$ for $1\!\le\!k\!<\!K$ and $I_K=[s_{K-1}, s_K]$ for the last interval. For mapping to the histogram, we can also mark $[s_{k-1}, s_k)$ as $[s_{k-1}, s_k-1]$. Segment fitting focuses on fitting the slope parameter $\beta_k (1\!\le\!k\!\le\!K)$ of the line segment for each interval. Although an intercept parameter of the PL function is also derived, we do not record it, only the slope parameter $\beta_k$ and the estimated interval frequency, i.e., the sum of frequencies of values in each interval, denoted as $\hat{f}_k\!=\!\sum_{v \in S_k} \hat{f}_v^{\text{H}} (1\!\le\!k\!\le\!K)$. These two parameters can be further optimized using the collected interval frequencies in subsequent phases. 

\textbf{Phase 2: PriPL-Tree Construction.}
Based on the PL function, we dynamically construct the PriPL-Tree structure in this phase, as detailed in Section~\ref{subsec:tree_construction}. Each leaf node corresponds to a fitted segment in sequence, e.g., $n_1$ to interval $I_1$ and $n_2$ to interval $I_2$, as shown in Figure~\ref{fig:framework} (a) and (b). Each non-leaf node represents an interval encompassing its children and has a non-uniform branch number (i.e., fan-out). This flexible structure is designed to minimize average error in responding to range queries. 

Given the PriPL-Tree structure, we allocate the remaining $N\cdot(1\!-\!\alpha)$ users to nodes and estimate their frequencies. Because the intervals of nodes along each path from the root to the leaves overlap, each user is randomly allocated to one node per path. As a result, the total number of users along each path is $N\cdot(1\!-\!\alpha)$. Each individual user is assigned multiple nodes with non-intersecting intervals that jointly cover the entire domain. Informed of these intervals, users can encode their values into bit vectors to employ the OUE mechanism with privacy budget $\epsilon$ for frequency estimation. For example, if a user's value $v$ is covered by nodes $\{n_3, n_7\}$ and he receives the intervals of nodes $(n_6, n_3, n_4, n_5)$, he can encode his value as $(0,1,0,0)$ and apply OUE. By aggregating all users' perturbed values for corresponding nodes, we derive each node's frequency $\bar{f}_k$, forming a preliminary PriPL-Tree, as shown in Figure~\ref{fig:framework}~(b). Each leaf node has two frequencies: $\hat{f}_k$, estimated during private PL fitting in phase 1, and $\bar{f}_k$, estimated by OUE in this phase. 

\textbf{Phase 3: PriPL-Tree Refinement.}
In the current PriPL-Tree, there are several frequency inconsistencies: (1) the estimated frequency of values or intervals may be beyond the actual range of $[0,1]$, (2) the frequency of a parent node may differ from the frequency sum of its child nodes, (3) two different frequencies occur at leaf nodes. To address these issues, we propose a post-processing method in Section~\ref{subsec:tree_refinement}, yielding an optimized PriPL tree as shown in Figure~\ref{fig:framework}~(c). It has consistent frequencies $\tilde{\mathbf{F}}$ across all nodes and optimized slopes $\tilde{\beta}_k$ for leaf nodes.
By now, a well-estimated PriPL-Tree is ready to respond to range queries. 

\textbf{Response to 1-D Range Query $Q([l,r])$.}
Given a 1-D range query $Q([l,r])$, the response is obtained by summing the frequencies $\tilde{f}_k$ of nodes $n_k$ that are fully within the range $[l,r]$ but whose parents are not, as well as frequencies from parts of leaf nodes that overlap but are not completely within $[l,r]$. For example, when querying $Q([200,1024])$ in Figure~\ref{fig:framework} (c), we aggregate the frequencies $\tilde{f}_2$ of node $n_2$, $\tilde{f}_7$ of node $n_7$, and the frequency of the sub-range $[200,340)$ (i.e., $[200,339]$) within node $n_1$. All these nodes and their corresponding frequencies can be derived by traversing the PriPL-Tree from top to bottom. Let $[l_{\text{sub}}, r_{\text{sub}}]$ denote the intersecting range of $[l,r]$ with the interval $I_k$ of a leaf node $n_k$; the frequency of this sub-range $Q([l_{\text{sub}}, r_{\text{sub}}])$ can be computed using Eq. (\ref{eq:sub_range}). The detailed computation process is shown in Appendix~A.1 in \cite{wang2024priplt}.  
\begin{equation}
	\textstyle
	\resizebox{0.441\textwidth}{!}{$
	Q([l_{\text{sub}}, r_{\text{sub}}])\!=\!
	\left(r_{\text{sub}}\!-\!l_{\text{sub}}\!+\!1\right)\cdot\left(\tilde{\beta}_k \left(\frac{l_{\text{sub}}+r_{\text{sub}}+1\!-\!|I_k|}{2}\!-\!s_{k\!-\!1}\right)\!+\!\frac{\tilde{f}_k}{|I_k|}\right)$}
	\label{eq:sub_range}
\end{equation}


\subsection{Private PL Fitting}
\label{subsec:pl_fitting}

As a fundamental data model, the PL function has been extensively studied in stream compression \cite{keogh2001online, buragohain2006space, liu2008novel, elmeleegy2009online, xie2014maximum} and learned index \cite{galakatos2019fiting, ferragina2020pgm, li2020lisa, li2021finedex} applications. In these contexts, the original data distribution is available, and the PL model can be learned using heuristic algorithms with specified error restrictions. However, our task poses a key challenge as we aim to use a noisy histogram to fit an unknown distribution while achieving an optimal error. To address this challenge, we carefully design the following segment fitting and interval partitioning steps to learn a PL model on the data distribution. The pseudocode is provided in Algorithm~\ref{algo:pl_fitting}.

\subsubsection{Segment Fitting}
Let's start with a simple case with $K$ partitioned intervals, and we aim to optimize the PL function by minimizing the squared error between the fitted and the noisy values. To alleviate the impact of LDP noise, we assume the PL function is continuous. This allows us to model the entire noisy histogram as a whole, leveraging all histogram data to fit each line segment, rather than using only a subset of data located in individual intervals, which may be overwhelmed by LDP noise. Moreover, this assumption is practical even for non-continuous distributions, as data around the breakpoints can be approximated by a continuous function with a sharp line connecting two breakpoints. Such approximation would produce only minor errors, especially for bucketized histogram values.

We model the continuous PL function in Eq.(\ref{eq:pl_func}) and illustrate it with $K\!=\!5$ segments in Figure~\ref{fig:framework} (a). Each $s_k$ ($0\!<\!k\!<\!K$) represents a breakpoint between two segments, while $s_0$ and $s_K$ mark the domain's endpoints. Let $\beta_0$ denote the intercept of the first segment. For each segment in the interval $I_k$, its slope is denoted by $\beta_k$, and its linear expression appears in the $k$-th row of Eq.(\ref{eq:pl_func}).
\begin{equation}
	f(v)\!=\!\left\{
	\begin{aligned}
		&\beta_0\!+\!\beta_1(v\!-\!s_0),\!&\!s_0\!\le\!v\!<\!s_1\\
		&\beta_0\!+\!\beta_1(s_1\!-\!s_0)\!+\!\beta_2(v\!-\!s_1),\!&\!s_1\!\le\!v\!<\!s_2\\
		&...&\\
		&\textstyle\beta_0\!+\!\sum_{k=1}^{K\!-\!1}\beta_k(s_k\!-\!s_{k\!-\!1})\!+\!\beta_K (v\!-\!s_{K\!-\!1}),\!&\!s_{K\!-\!1}\!\le\!v\!\le\!s_{K}
	\end{aligned}
	\right.
	\label{eq:pl_func}
\end{equation}

Given the noisy histogram with frequency $\hat{f}_v^{\text{H}}$ for $v \in [d]$, our objective is to minimize the squared error between the PL function $f(v)$ and the observed frequencies, i.e., $\min \sum_{v\in[d]}(f(v) - \hat{f}_v^{\text{H}})^2$.

For ease of optimization, we express the problem using matrices. Let $\mathds{1}$ represent an indicator function, which is 1 when its predicate is met and 0 otherwise. We denote the frequencies in the noisy histogram by the vector $\mathbf{\hat{F}}^{\text{H}} = [\hat{f}_1^{\text{H}} \ \hat{f}_2^{\text{H}} \ \ldots \ \hat{f}_d^{\text{H}}]^T$, and the parameters of the PL function by $\mathbf{B} = [\beta_0 \ \beta_1 \ \ldots \ \beta_K]^T$. We define two matrices, $\mathbf{X}_{d\times(K+1)}$ and $\mathbf{A}_{d\times(K+1)}$, as described in Eq.(\ref{eq:x_matrix}) and Eq.(\ref{eq:a_matrix}) respectively. In both matrices, the $v$-th ($v\in[d]$) row corresponds to the expression of $f(v)$ by Eq.(\ref{eq:pl_func}). For $k > 0$, the $k$-th column in $\mathbf{X}$ refers to the term $(v-s_{k-1})$, and the $k$-th column in $\mathbf{A}$ refers to the term $(s_{k}-s_{k-1})$.
\begin{equation}
\textstyle
	\mathbf{X} = 
	\begin{bmatrix}
		1 & (1-s_0)\mathds{1}_{s_0\le 1 < s_1} & \ldots &(1-s_{K-1})\mathds{1}_{s_{K-1} \le 1 \le s_{K}}\\
		1 & (2-s_0)\mathds{1}_{s_0\le 2 < s_1} & \ldots &(2-s_{K-1})\mathds{1}_{s_{K-1} \le 2 \le s_{K}}\\
		\vdots & \vdots & \ddots & \vdots\\
		1 & (d-s_0)\mathds{1}_{s_0\le d < s_1} & \ldots &(d-s_{K-1})\mathds{1}_{s_{K-1} \le d \le s_{K}}
	\end{bmatrix}
	\label{eq:x_matrix}
\end{equation}
\begin{equation}
\textstyle
	\mathbf{A} = 
	\begin{bmatrix}
		0 & (s_1-s_0)\mathds{1}_{1 > s_1} & \ldots & (s_{K-1}-s_{K-2})\mathds{1}_{1>s_{K-1}} & 0 \\
		0 & (s_1-s_0)\mathds{1}_{2 > s_1} & \ldots & (s_{K-1}-s_{K-2})\mathds{1}_{2>s_{K-1}} & 0 \\
		\vdots & \vdots & \ddots & \vdots & \vdots \\
		0 & (s_1-s_0)\mathds{1}_{d > s_1} & \ldots & (s_{K-1}-s_{K-2})\mathds{1}_{d>s_{K-1}} & 0 \\
	\end{bmatrix}
	\label{eq:a_matrix}
\end{equation}

By calculating $(\mathbf{X} + \mathbf{A})\cdot\mathbf{B}$, we derive the PL fitted frequencies for all values in $[d]$, i.e., $[f(1) \ f(2) \ \ldots \ f(d)]^T$. Consequently, we reformulate the optimization problem using matrices as follows:
\begin{equation}
	\min ((\mathbf{X}+\mathbf{A})\cdot\mathbf{B}-\mathbf{\hat{F}}^{\text{H}})^T\cdot((\mathbf{X}+\mathbf{A})\cdot\mathbf{B}-\mathbf{\hat{F}}^{\text{H}}).
	\label{eq:opt_matrix}
\end{equation}

Let $L(\mathbf{B}) = ((\mathbf{X}+\mathbf{A})\cdot\mathbf{B}-\mathbf{\hat{F}}^{\text{H}})^T\cdot((\mathbf{X}+\mathbf{A})\cdot\mathbf{B}-\mathbf{\hat{F}}^{\text{H}})$ denote the loss function. By setting $\partial L(\mathbf{B}) / \partial \mathbf{B} = 0$, we derive the closed-form solution for $\hat{\mathbf{B}}$, where $\hat{B}_0$  represents the estimated intercept $\hat{\beta}_0$ and $\hat{B}_k$ ($k>0$) represents the estimated slope $\hat{\beta}_k$ of the $k$-th segment.
\begin{equation}
	\hat{\mathbf{B}} = ((\mathbf{X}+\mathbf{A})^T(\mathbf{X}+\mathbf{A}))^{-1}(\mathbf{X}+\mathbf{A})^T\mathbf{\hat{F}}^{\text{H}}.
	\label{eq:solution_A}
\end{equation}

\begin{algorithm}[tbp]
	\small
	\caption{Private PL Fitting}
	\label{algo:pl_fitting}
	\KwIn{Noisy histograms $\hat{\mathbf{F}}^{\text{EM}}$ and $\hat{\mathbf{F}}^{\text{EMS}}$ by SW mechanism.}
	\KwOut{A PL function with adaptive $K$ segments.}
	Set segment number $K\!=\!1$ and breakpoints $\mathbf{S}\!=\!(0,d\!-\!1)$\;
	Initialize search space $\Theta = \{s | s\in[d-1]\}$\;
	\For{$i \in \{1,2\}$}{
		\leIf{$i == 1$}{$\hat{\mathbf{F}}^{\text{H}} = \hat{\mathbf{F}}^{\text{EM}}$}{$\hat{\mathbf{F}}^{\text{H}} = \hat{\mathbf{F}}^{\text{EMS}}$}
		\While{$RSS$ not converges or $K \le (K_{\text{max}} \times i / 2)$}{
			$K = K+1$\;
			\ForEach{Breakpoint candidate $s$ in $\Theta$}{
				Initialize $\mathbf{X}_{d\times(K+1)}$ and $\mathbf{A}_{d\times(K+1)}$ with $\mathbf{S}$ and $s$\;
				Calculate PL parameters $\hat{\mathbf{B}}$ with Eq.(\ref{eq:solution_A})\;
				Calculate $RSS_k\!=\!\sum_{v\in I_k} (\hat{f}_v^{\text{H}}\!-\!f(v))^2$ for each interval $S_k$, record total $RSS = \sum_{1\le k\le K}RSS_{k}$\;
			}
			\textls[-29]{Append $s^*$ to $\mathbf{S}$, where $s^*\!\in\!\Theta$ produces the minimal $RSS$}\;
			\textls[-20]{Set $\Theta= \{s |s\!\in\!I_{k^*}\}$, where $I_{k^*}$ has the maximum $RSS_{k^*}$ and its frequency $\hat{f}_k\!=\!\sum_{v\in I_k} \hat{f}_v^{\text{EMS}}\!>\!\sigma\sqrt{(1\!-\!\alpha)}$}\; 
		}
		
	}
	\KwRet \textls[-10]{A PL function with breakponts $\mathbf{S}$, slopes $\hat{\mathbf{B}}$ and frequencies $\hat{\mathbf{F}}$} \;
\end{algorithm}

\subsubsection{Interval Partitioning}
\label{subsubsec:interval_partitioning}

Building upon segment fitting, we propose a greedy method to search breakpoints one by one, achieving approximately optimal interval partitions. As described in Algorithm~\ref{algo:pl_fitting}, during each search (i.e., each iteration in lines 5$\sim$12), we traverse all candidate breakpoints $s$ in the search space $\Theta$, fit segments based on it and the existing breakpoints $\mathbf{S}$, and find the best breakpoint $s^*$ that minimizes the residual sum of squares (RSS). The initial search space contains all possible candidates in the whole domain. In subsequent iterations, we select values in the interval $I_{k*}$ as the new search space. This interval $I_{k*}$ has the maximum RSS for its fitted line segment, indicating it requires further splitting for a more accurate fit. Additionally, its frequency $\hat{f}_k$ should be greater than $\sigma\sqrt{(1-\alpha)}$, ensuring that it will not be overwhelmed by OUE noise in node frequency estimation in the next phase. This iteration continues until the maximum number of segments is reached (we set $K_{max}=32$ in experiments) or RSS converges, i.e., the ratio of total RSS between two consecutive iterations approaches 1, indicating no further gain from increasing segments. 

Additionally, we propose two strategies to improve interval partitioning for effectiveness and efficiency. These are briefly described below, with detailed explanations provided in Appendix~A in \cite{wang2024priplt}.

(1) \emph{Twice Partitioning Strategy for Effectiveness}: To fit both jagged and smooth distributions, we sequentially perform interval partitioning on two distributions, $\hat{\mathbf{F}}^{\text{EM}}$ and $\hat{\mathbf{F}}^{\text{EMS}}$, as outlined in lines 3$\sim$4 of Algorithm~\ref{algo:pl_fitting}. These distributions, derived from SW using EM and EMS, respectively, represent an initially calibrated (typically jagged) distribution and a smoothed one with reduced noise \cite{li2020estimating}. After this process, we derive the necessary partitions to depict both types of distributions. Notably, these two estimations from SW require only one perturbation per user, not increasing the privacy budget. 

(2) \emph{Search Acceleration Strategy for Efficiency}: To accelerate the search process, we propose a multi-granular search strategy instead of traversing all possible breakpoints at each search in line~7 in Algorithm~\ref{algo:pl_fitting}. Given the granularity factor $\phi$, which limits the maximum number of candidate breakpoints during each search, we initially explore the space $\Theta$ using a step of $\left\lceil |\Theta|/\phi \right\rceil$ to identify an optimal breakpoint $s^*$. Subsequently, we narrow the search space to $\left[s^*\!-\!\left\lceil |\Theta|/\phi \right\rceil, s^*\!+\!\left\lceil |\Theta|/\phi \right\rceil \right]$ and search it with a finer step of $\left\lceil |\Theta|/\phi^2 \right\rceil$. We repeat this process until the step size reduces to 1. A relatively smaller $\phi$ ($\phi < d$) accelerates the search while increasing the probability of encountering local optima. Since breakpoints inherently represent different local optima partitioning the domain, $\phi$ primarily influences the order of breakpoint discovery rather than the final interval partitions. For brevity, this strategy is not included in Algorithm~\ref{algo:pl_fitting}, but it can replace line 7 of it.

\color{black}

\subsection{PriPL-Tree Construction}
\label{subsec:tree_construction}

Based on the derived partitioned intervals, we construct the PriPL-Tree, focusing on user allocation and tree structure construction and providing the pseudocode in Algorithm~\ref{algo:pripl_tree_construction}.

\begin{algorithm}[tbp]
\small
	\caption{PriPL-Tree Construction}
	\label{algo:pripl_tree_construction}
	\KwIn{$K$ segments and $N(1-\alpha)$ users}
	\KwOut{PriPL-Tree $\mathcal{T}$}
	\tcp{\small Tree Structure Construction}
	Construct a basic balanced binary tree $\mathcal{T}$\;
	\For{node $n_k$ in postorder traversal}{
		Compute $Err$ for $\mathcal{T}$ with $n_k$ and $Err'$ w/o $n_k$ using Eq.(\ref{eq:exp_error})\;
		\lIf{$Err' < Err$}{
			Remove $n_k$
		}
	}
	\tcp{\small User Allocation}
	Assign all available users $U_0\!=\!\{u_1, u_2, \ldots, u_{N(1-\alpha)}\}$ to the root\;
	\For{node $n_k$ in level order traversal}{
		\lIf{$n_k$ is root}{
			Set $U_k' = \phi$
		}
		\Else{
			Compute $h_k$, the height of the subtree rooted at $n_k$\;
            Set $U_k'$ as $\left\lceil\frac{|U_k|}{h_k}\right\rceil$ randomly sampled users from $U_k$\;
            Allocate $U_k'$ to $n_k$ for frequency estimation\;
		}
		\lFor{node $n_c \in children(n_k)$}{
			Assign $U_c\!=\!U_k\!-\!U_k'$ to $n_c$
		}
	}
	\tcp{\small Node Frequency Estimation}
	\For{each user $u_i \in U_0$}{
		Assign intervals of nodes $\mathcal{N}_i = \{n_k | u_i \in U_k'\}$ to $u_i$\;
		Collect OUE perturbed vectors with size $|N_i|$\;
	}
	Estimate frequency $\bar{f}_k$ for each node $n_k$ based on OUE\;
	\KwRet PriPL-Tree $\mathcal{T}$\;
\end{algorithm}

\subsubsection{User Allocation}

Given a potentially unbalanced PriPL-Tree $\mathcal{T}$, we explore user allocation strategies. Nodes along each path from the root to the leaves have overlapping intervals and related frequencies, prompting us to allocate users to nodes along paths \cite{wang2023privnud} rather than by layers \cite{cormode2019answering, du2021ahead, wang2019answering}. The allocation process is detailed in lines 5$\sim$12 of Algorithm~\ref{algo:pripl_tree_construction}. Initially, all unallocated users are assigned to the root (line 5). Because the root has a constant frequency of $1$ and requires no estimation, it is allocated no users (line 7) and just passes the user set to its children. For each non-root node $n_k$ that we traversed in level order (line 6), it has inherited the unallocated user set $U_k$ from its parent (line 12), uniformly samples $1/h_k$ of these users for itself, marked as $U_k'$ (lines 10$\sim$11), and passes the remaining users $U_k - U_k'$ to its children (line 12). Here, $h_k$ represents the height of the subtree rooted at $n_k$, ensuring uniform allocation along the longest path. During this process, all child nodes of $n_k$ will receive the same to-be-allocated user group $U_k - U_k'$ since they do not overlap in their intervals.

An example of user allocation is shown in Figure~\ref{fig:tree_construction}, where colored rectangles above each node represent the randomly assigned users. Along each path, such as ``$n_0-n_7-n_8-n_3$'' in Figure~\ref{fig:tree_construction} (a), the total number of users is $N'$. For each user, like the one represented by the yellow rectangle, he will participate in frequency estimations for multiple nodes, such as $\{n_6, n_3, n_4, n_5\}$. The intervals of these nodes do not intersect and collectively cover the entire domain.

\begin{figure}[htb]
\centering
\includegraphics[width=0.48\textwidth]{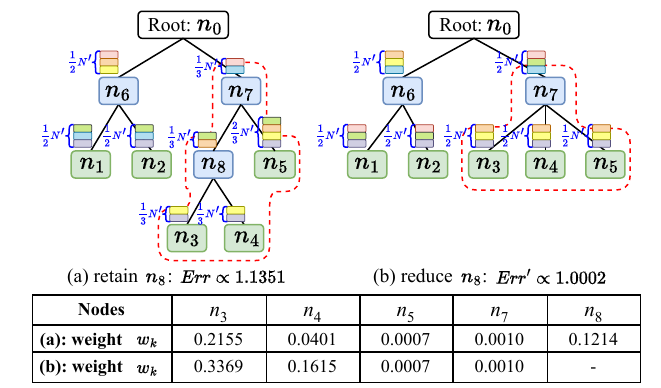}
\caption{An Example of Tree Construction \rm{($N'= N(1-\alpha)$)}}
\label{fig:tree_construction}
\end{figure}

\subsubsection{Tree Structure Construction}

Initially, we can construct a basic balanced binary tree, as illustrated in Figure~\ref{fig:tree_construction} (a). For optimization, we perform \emph{adaptive node reduction}, where the reduction of a node refers to removing it from the tree and linking its child nodes to its parent. For example, reducing node $n_8$ in Figure~\ref{fig:tree_construction} (a) produces the tree in Figure~\ref{fig:tree_construction} (b). We examine all non-leaf nodes through postorder traversal, adaptively determining whether to reduce each node to minimize the average error in response to range queries. 

Specifically, the average error for all possible queries can be evaluated by accumulating the noise error from nodes within query ranges and the PL fitting error from intersecting leaf nodes. Since the PL fitting error primarily depends on estimates from phase 1 (which are determined), we focus on the noise error of nodes, as formalized in the left part of Eq.(\ref{eq:exp_error}). Here, $\mathcal{N}$ denotes all nodes in the PriPL-tree, $w_k$ denotes the probability of node $n_k$ being involved in queries, and $\alpha_k$ denotes the ratio of allocated users for node $n_k$. As the variance of OUE used in node frequency estimation is inversely proportional to its allocated user ratio, we simplify $Err$ to the right part of Eq.(\ref{eq:exp_error}). Assuming all queries arrive with equal probabilities, the weight $w_k$ is calculated by Eq. (\ref{eq:weight}) \cite{qardaji2013understanding}, where $[l_k,r_k]$ ($[l_p, r_p]$) denotes the interval of node $n_k$ (its parent $n_p$), and the total number of possible queries is $(d\!+\!1)d/2$. 
\begin{gather}
	\textstyle
	Err = \sum_{n_k\in \mathcal{N}} w_k \cdot Var(\bar{f}_k) \propto \sum_{n_k\in \mathcal{N}} w_k / \alpha_k 
	\label{eq:exp_error} \\
	w_k = \left(l_k \cdot (d - r_k + 1) - l_p \cdot (d-r_p+1)\right)\big/\left((d+1)d/2\right)
	\label{eq:weight}
\end{gather}

We provide an example in Figure~\ref{fig:tree_construction} to calculate errors and decide whether to reduce $n_8$. Considering that the existence of a non-leaf node only influences user allocation and weight calculation for its ancestor and descendant nodes, as framed by a red dash line in Figure~\ref{fig:tree_construction}, we compare the accumulated error from these nodes rather than the entire tree. Finally, we reduce $n_8$ for a smaller error. 

\subsection{PriPL-Tree Refinement}
\label{subsec:tree_refinement}

To address the frequency inconsistencies outlined in Section~\ref{subsec:workflow}, we perform frequency and slope refinements as follows.

\subsubsection{Frequency Refinement}

As we know, several methods can address these inconsistencies individually, such as Norm-Sub for issue (1) \cite{wang2020locally} and constrained inference for issue (2) \cite{hay2009boosting, qardaji2013understanding}. However, applying one method may lead to the emergence of another inconsistency. AHEAD \cite{du2021ahead} employs these two techniques iteratively to solve issues (1) and (2), but this is inefficient and often inaccurate. To address all three inconsistency issues simultaneously, we devise an optimized constrained inference method comprising two steps: \emph{weighted averaging} and \emph{frequency consistency}, each requiring only a single tree traversal.

In the \emph{weighted averaging} step, we traverse nodes from leaves to root, minimizing frequency variance for each node and addressing inconsistency issue (3). For a leaf node $n_k$ with two frequency estimates, $\hat{f}_k$ (from node frequency estimation) and $\bar{f}_k$ (from the noisy histogram estimation), we update its frequency to $\dot{f}_k\!=\!\theta \hat{f}_k\!+\!(1\!-\!\theta) \bar{f}_k$, where $\theta\!=\!\var(\bar{f}_k)/(\var(\hat{f}_k)\!+\!\var(\bar{f}_k))$, achieving the minimal variance $\var(\dot{f}_k)\!=\!\var(\hat{f}_k)\var(\bar{f}_k)/(\var(\hat{f}_k)\!+\!\var(\bar{f}_k))$. For non-leaf node $n_k$, we similarly update its frequency to $\dot{f}_k$, using its frequency $\hat{f}_k$ and the sum of its child nodes' frequencies $\sum_{n_c\in child(n_k)}\hat{f}_c$. 

In the \emph{frequency consistency} step, we update frequencies from the root to leaves to address inconsistency issues (1) and (2). The root's frequency is fixed at 1, i.e., $\tilde{f}_{\text{root}}\!=\!1$. Given a parent node $n_p$ with optimized frequency $\tilde{f}_p\!\ge\!0$ and its child nodes' to-be-optimized frequencies $\{\dot{f}_{c} | n_c\!\in\!\text{child}(p)\}$, we define the optimization problem in Eq.(\ref{eq:ms_obj}). Let $D_+$ be the set of child nodes with positive updated frequencies, and $D_0$ be those with zero updated frequencies. According to KKT condidtions, the optimal frequency is $\tilde{f}_{c}\!=\dot{f}_{c}\!+\!\left(\tilde{f}_{p}\!-\!{\sum_{n_c\in D_+}\dot{f}_{c}}\right)\big/|D_+|$ if $n_c\!\in\!D_+$, and $\tilde{f}_{c}\!=\!0$ if $n_c\!\in\!D_0$.
\begin{equation}
\begin{gathered}
\textstyle
	\min \sum_{n_c\in child(n_p)}(\tilde{f}_{c}- \dot{f}_{c})^2 
\\ \textstyle
	\text{s.t.} \quad \sum_{n_c\in child(n_p)}\tilde{f}_c = \tilde{f}_p\ \ and \ \ 
	\forall n_c\in child(n_p),\tilde{f}_{c} \geq 0
	\label{eq:ms_obj}
\end{gathered}
\end{equation}

\subsubsection{Slope Refinement}
To meet frequency constraints within nodes, i.e., resolving inconsistency issue (1), we refine slopes based on optimized node frequencies. For node $n_k$ with frequency $\tilde{f}_k$, we must guarantee non-negativity at both endpoints of interval $I_k$. Denoting the endpoints of $I_k$ as $l_k$ and $r_k$, we require $f(l_k) = \tilde{f}_k/|I_k| - \tilde{\beta}_k (|I_k|-1)/2 \geq 0$ and $f(r_k) = \tilde{f}_k/|I_k| + \tilde{\beta}_k (|I_k|-1)/2 \geq 0$. This establishes an effective range $[-C_k, C_k]$ for its slope, where $C_k = 2\tilde{f}_k/|I_k|(|I_k|-1)$. We then update $\tilde{\beta}_k$ using Eq.(\ref{eq:refine_slopes}), minimizing the error $(\tilde{\beta}_k-\hat{\beta}_k)^2$ and ensuring all fitted frequencies in $n_k$ are non-negative.
\begin{equation}
	\tilde{\beta_k}\!=\!\left\{
	\begin{aligned}
		&-C_k, & \hat{\beta_k} < - C_k\\
		&\hat{\beta_k}, &-C_k \le \hat{\beta_k} \le C_k\\
		&C_k, &\hat{\beta_k} > C_k
	\end{aligned}
	\right.
	\label{eq:refine_slopes}
\end{equation}

\subsection{Privacy and Performance Analysis}

In this subsection, we analyze the privacy guarantee, estimation error, and space and time complexity of PriPL-Tree for range queries.

\subsubsection{Privacy Analysis}
During the PriPL-Tree estimation, we collect and estimate frequencies twice from users via SW and OUE, each employing non-overlapping subsets of users and utilizing the full privacy budget $\epsilon$. The PL-fitting, PriPL-Tree construction (excluding node estimation), and refinement are post-processing steps over these collected data. As such, our PriPL-Tree satisfies $\epsilon$-LDP.

\subsubsection{Error Analysis}
\label{subsubsec:error_analysis}
Range query errors from PriPL-Tree arise from two sources: \emph{noise and sampling error} and \emph{PL estimation error}.

\emph{Noise and sampling error} arise from the frequency estimation via LDP mechanisms using a subset of users. In the private PL fitting phase (phase 1), the estimated frequency $\hat{f}_k$ of leaf node $n_k$, derived via the SW mechanism \cite{li2020estimating}, exhibits bias depending on the data distribution and has a square error of $O(\frac{|I_k|}{N\alpha \epsilon^2})$ \cite{duan2022utility}. In the PriPL-Tree construction phase (phase 2), the node frequency $\bar{f}_k$, estimated via the OUE mechanism \cite{wang2017locally}, is unbiased and has variance $\var(\bar{f}_k) = \frac{4e^\epsilon}{ \alpha_k N\cdot(e^{\epsilon}-1)^2} = O(\frac{1}{ \alpha_k N\epsilon^2})$, where $\alpha_k$ represents the proportion of users allocated to node $n_k$. During the PriPL-Tree refinement (phase 3), we aggregate all these frequency estimates to minimize the variance of each node's frequency, leading to error bounds presented in Theorem~\ref{theorem:pripl_tree_node_error}. 
For the precise square error needed for multi-dimensional grid consistency refinement (Section~\ref{subsec:consistency_refinement}), we introduce a numerical method. We treat the updated frequency after refinement as a weighted average of nodes' frequency estimates from the first two phases. By accounting for specific weights, we can derive accurate errors. Due to space constraints, we elaborate on this numerical method, and the proof of Theorem~\ref{theorem:pripl_tree_node_error} in Appendix~C in \cite{wang2024priplt}.

\begin{theorem}
\label{theorem:pripl_tree_node_error}
	Given a PriPL-Tree with at most $K$ segments (corresponding to $K$ leaf nodes), the error variance of frequencies after weight averaging in refinement (phase~3) is $O\left(\frac{K\cdot \log K}{(1-\alpha) \cdot (K+1) \cdot N \cdot \epsilon^2}\right)$ for non-leaf nodes and $O\left(\frac{\log K}{(1-\alpha) \cdot N \cdot \epsilon^2}\right)$ for leaf nodes. After frequency consistency, these variance is capped at $O\left(\frac{K\cdot\log K}{(1-\alpha)\cdot N \cdot \epsilon^2}\right)$.
\end{theorem}

\emph{PL estimation error} arises when estimating the frequency of the sub-range $[l_{\text{sub}}, r_{\text{sub}}]$ within a leaf node $n_k$ under a linear assumption. For each PL estimated frequency $f(v)$ of value $v$ in this subrange, its square error can be approximated by $\E((f(v)-f_v)^2) = \E((f(v)-\hat{f}_v^{\text{H}}) + (\hat{f}_v^{\text{H}}-f_v))^2 \le 2 (\E(f(v)-\hat{f}_v^{\text{H}})^2 + \E(\hat{f}_v^{\text{H}}-f_v)^2)$. The first term represents the square error of the PL function fitting the noisy histogram, while the second term represents the noise error of the noisy histogram. The magnitude of this error depends on the noisy histogram's distortion degree, the segment number, and the actual data distribution. It tends to be small for smooth distributions and large for jagged distributions. Ultimately, the total error of $Q([l_{\text{sub}}, r_{\text{sub}}])$ accumulates the error of values in it, leading to a result proportional to the range size square $(r_{\text{sub}} - l_{\text{sub}} + 1)^2$.

\subsubsection{Space and Time Complexity Analysis}
\label{subsec:time_complexity}
Assuming PriPL-Tree's maximum segment number, i.e., the number of leaf nodes, is $K$, the space complexity includes the size of PriPL-Tree, $O(K)$, and the size of parameter matrices $\mathbf{X}$ and $\mathbf{A}$ for private PL fitting, $O(K\cdot d)$, which is $O(K\cdot d)$ in total. 
The time complexity of PriPL-Tree involves two parts --- the construction time and the query time. Overall, the construction time complexity mainly arises from private user data aggregation, frequency estimation, and private PL fitting during phases 1 and 2, totaling $O(N\cdot K + d\cdot T + d\cdot\log d\cdot K^3)$. Here, $T$ denotes the number of iterations for EM and EMS in the SW mechanism during distribution estimation in phase 1. The query time complexity, proportional to the tree height, is $O(\log_2 K)$. Due to space limitations, we provide a detailed analysis of the time complexity for each phase in Appendix~E.1 of \cite{wang2024priplt}.

\section{Extension with Adaptive Grids for Multi-Dimensional Queries}
\label{sec:adaptive_grid}

Building on insights from existing works summarized in Section~\ref{subsec:summary}, we combine 1-D PriPL-Trees and 2-D grids to handle multi-dimensional range queries.
In contrast to uniformly constructed 2-D grids proposed by HDG \cite{yang2020answering}, we introduce data-aware adaptive grids. These grids leverage the accurate 1-D marginal distributions from PriPL-Trees to dynamically partition the entire domain into dense or sparse cells, adapting to data distribution density. As a result, they offer a more precise representation of 2-D data distributions and more accurate range query responses.

In this section, we first outline the workflow for multi-dimensional range queries and then present the core methods of \emph{adaptive grid partitioning} and \emph{consistency refinement}. Due to space limitations, we analyze the estimation error and runtime complexity in Appendix~D and Appendix~E.2 in our full version of paper \cite{wang2024priplt}, respectively.

\subsection{Workflow of Multi-dimensional Cases}
\label{subsec:workflow_md}

We list the workflow of multi-dimensional range queries below and provide a figure illustration in Appendix~B.1 of \cite{wang2024priplt}.

\textbf{Step 1: Estimating $m$ 1-D Histograms using PriPL-Tree.} 
Initially, we allocate half of the users to estimate 1-D marginal distributions.
For each attribute $A_i$, we employ $N/2m$ users to estimate the PriPL-Tree and generate histograms $\tilde{\mathbf{F}}_i$ using the PL function within the PriPL-Tree. These histograms enable us to depict each attribute's underlying data distribution effectively.

\textbf{Step 2: Estimating $\binom{m}{2}$ 2-D Adaptive Grids.} 
For each attribute pair $\langle A_i, A_j \rangle$ ($i,j \in [m]$ and $i \neq j$), we construct a 2-D grid, as presented in Section~\ref{subsec:adap_grid}. During each grid's construction, we partition the domain of each dimension into non-uniform intervals based on the marginal distribution, thereby forming 2-D grids that are denser in high-frequency regions and sparser in low-frequency regions, as depicted in Figure~\ref{fig:adaptive_grids}. After construction, we assign $N/2\binom{m}{2}$ users to estimate the frequencies of cells in each grid using the OUE mechanism with a privacy budget of $\epsilon$.

\textbf{Step 3: Refining Consistency Between Grids and PriPL-Trees.} 
Due to the noise introduced by the LDP mechanism, frequency inconsistencies for one attribute may arise among grids and PriPL-Trees, and some frequencies could be negative. To address these issues, we propose a post-processing method, detailed in Section~\ref{subsec:consistency_refinement}, optimizing the frequencies of grids and adjusting both the frequencies and slopes of nodes in PriPL-Trees.

\textbf{Step 4: Answering Range Queries.}
For 1-D queries, we can directly utilize PriPL-Trees to respond. 
For a 2-D query involving attributes $\langle A_i, A_j \rangle$, we answer it by a response matrix with $d_i \times d_j$ values. This matrix represents the 2-D data distribution and is estimated from 1-D histograms and 2-D adaptive grids using the maximum entropy algorithm or weighted update as described in \cite{yang2020answering, wang2023privnud}. The critical difference between us and \cite{yang2020answering, wang2023privnud} is that they enforce the frequency sum of a sub-region equal to match its corresponding 1-D cell, while we enforce each frequency in the marginal distribution of the matrix to match the 1-D histogram. This fine-grained consistency fully exploits the slope information in PriPL-Trees, yielding a more accurate distribution. 
Further, for $\lambda$-D ($\lambda > 2$) queries, we estimate with a $2^\lambda$ response matrix based on associated $2$-D queries as in \cite{yang2020answering, wang2023privnud}. 

\begin{figure}[htb]
\centering
\includegraphics[width=0.4\textwidth]{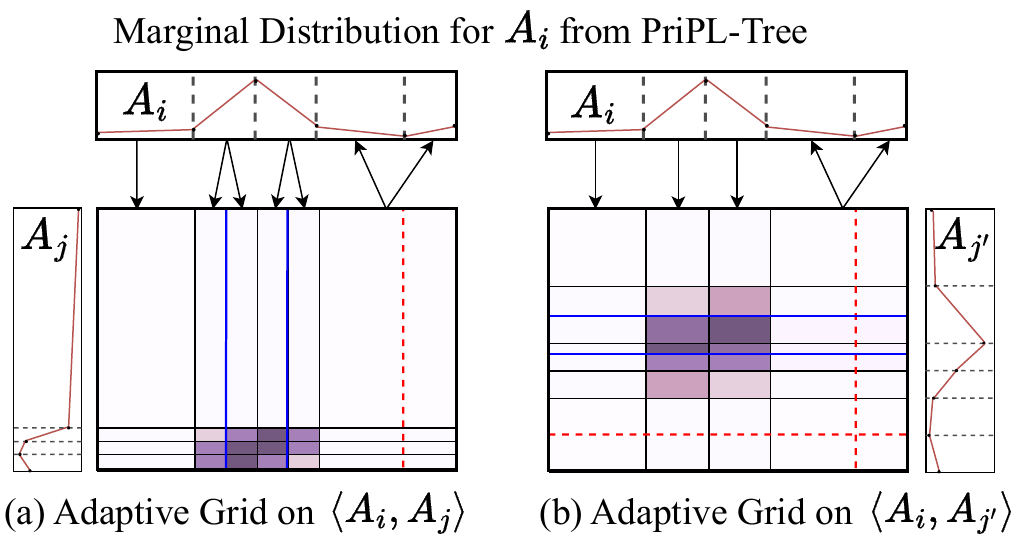}
\caption{Examples of Adaptive Grids \rm{(Black solid lines represent partitions inherited from PriPL-Trees; blue solid lines indicate newly added partitions; red dashed lines indicate deleted partitions.)}}
\label{fig:adaptive_grids}
\end{figure}

\subsection{Adaptive 2-D Grid Partitioning}
\label{subsec:adap_grid}

The 2-D grid depicts the underlying data distribution by assuming uniform frequency distribution within each cell. To enhance its data depiction capability, we partition the grid based on data density. In data-dense areas, densely partitioned cells provide a more accurate distribution representation, reducing reliance on uniform assumptions. Conversely, in data-sparse areas, low cell frequencies may be overwhelmed by OUE noise, diminishing their effectiveness. Therefore, in such areas, sparsely partitioned cells covering larger areas with relatively higher frequencies are preferable. 
Building on this concept, we introduce adaptive partitioning for 2-D grids, as illustrated in Figure~\ref{fig:adaptive_grids}. For the attribute pair $\langle A_i, A_j \rangle$, we initially partition its domain according to the leaf node partitions in the PriPL-Trees, creating an initial grid $G$ with $g_i \times g_j$ cells. We then dynamically adjust the partition lines, adding lines in high-frequency regions and removing them in low-frequency regions along each dimension, as indicated by the blue and red lines in Figure~\ref{fig:adaptive_grids}. The goal is to minimize the squared error $Err_G$ for range queries within the grid while ensuring each cell’s frequency exceeds the standard deviation of the OUE noise, $\sqrt{2\sigma^2 \cdot \binom{m}{2}}$.

We introduce the squared error $Err_G$ as below and detail the algorithm for adaptive partitioning in Appendix~B.2 in \cite{wang2024priplt}. For a cell $c$ in grid $G$, when fully covered by a query rectangle, it incurs a squared noise and sampling error of $2\sigma^2 \binom{m}{2}$. If the cell intersects with a query rectangle, it incurs an estimation error proportional to the intersecting frequency $f_c$, which can estimated by the product of its marginal frequencies. Let $\pi_{i}(\cdot)$ denotes the projection of the cell index on marginal attribute $A_i$, the squared estimation error is thus $\eta\cdot (\tilde{f}_{\pi_i(c)} \cdot \tilde{f}_{\pi_j(c)})^2$, with $\eta$ as a constant. 
 For a range query $Q$ selecting a portion $r$ of the area of a grid with $g_i \times g_j$ cells, the total squared error combines noise and sampling errors from $r g_i g_j$ cells and estimation errors proportional to $r$ times the square of all cells' frequencies, totaling $2 rg_i g_j \sigma^2 \binom{m}{2} + r\eta \sum_{c \in G}(\tilde{f}_{\pi_i(c)} \cdot \tilde{f}_{\pi_j(c)})^2$. Our experiments demonstrate that an $\eta$ value of 0.04 provides accurate estimations across various datasets.

\subsection{Consistency Refinement}
\label{subsec:consistency_refinement}

For partitions of $A_i$ between the PriPL-Tree and the grid for $\langle A_i, A_j \rangle$, we observe several one-to-many relationships, as shown by arrows in Figure~\ref{fig:adaptive_grids}. We treat each one-to-many relationship as a tree and can apply our optimized constrained inference method, detailed in Section~\ref{subsec:tree_refinement}, to ensure their frequency consistency and non-negativity. 
In a global view, each attribute $A_i$ is linked to $m-1$ grids, and each grid on $\langle A_i,A_j \rangle$ associates with two attributes. Applying the above method straightforwardly to update $A_i$ and each related grid sequentially is challenging for in maintaining global consistency. For instance, resolving consistency between $A_i$ and $\langle A_i, A_{j'} \rangle$ might reintroduce inconsistencies between $A_i$ and $\langle A_i, A_j \rangle$ that were previously resolved. Therefore, we first update all 1-D attributes, i.e., the leaf node frequencies in PriPL-Trees, by applying improved constrained inference sequentially across their $m-1$ corresponding grids. Using these updated 1-D frequencies, we then update the grids with the frequency consistency operation (i.e., the second step in the improved constrained inference). 
When multiple marginal cells in grids correspond to a single leaf node, we directly use the frequency of this leaf node to update the relevant cells. Conversely, if a marginal cell in grids corresponds to multiple leaf nodes, we update the cells' frequencies using the sum of frequencies from these leaf nodes. For a grid on $\langle A_i, A_j \rangle$, to prevent reintroducing inconsistencies with $A_i$ after aligning with $A_j$, we alternately update it with $A_i$ and $A_j$ until convergence.

Moreover, using the updated frequencies of each attribute, we can further refine the PriPL-Tree to enhance accuracy for 1-D range queries. The leaf node slope can be updated based on these frequencies as described in Section~\ref{subsec:tree_refinement}, and non-leaf node frequencies can be updated by aggregating the frequencies of their child nodes.

\section{Evaluation}
\label{sec:experiments}

In this section, we evaluate the performance of PriPL-Tree and its extension for both 1-D and multi-D range queries.

\begin{figure*}
	\centering
	\includegraphics[width=0.35\textwidth]{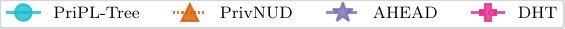}
	\vspace{-0.1in} 
	\\
	\subfloat[Gaussian]{
		\includegraphics[width=0.24\textwidth]{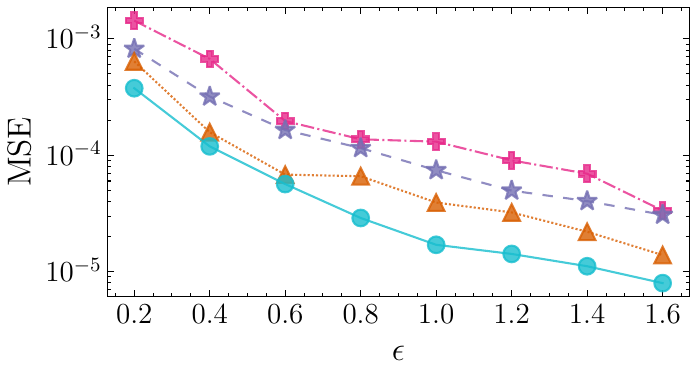}
	}
	\subfloat[MixGaussian]{
		\includegraphics[width=0.24\textwidth]{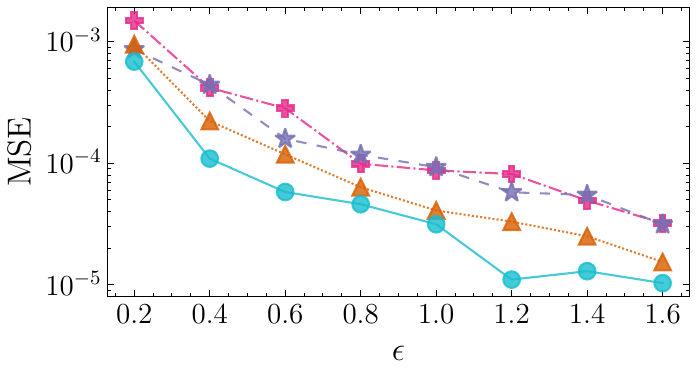}
	}
	\subfloat[Cauchy]{
		\includegraphics[width=0.24\textwidth]{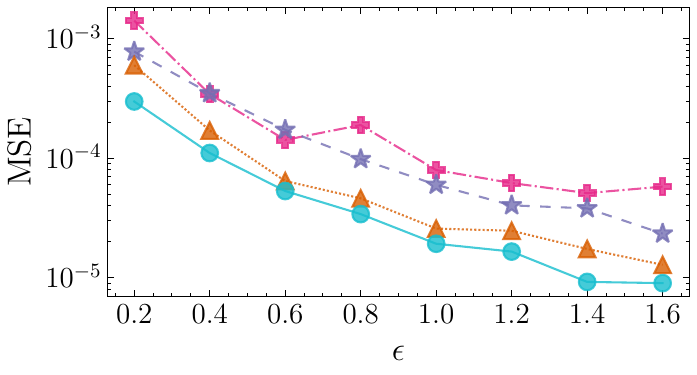}
	}
	\subfloat[Zipf]{
		\includegraphics[width=0.24\textwidth]{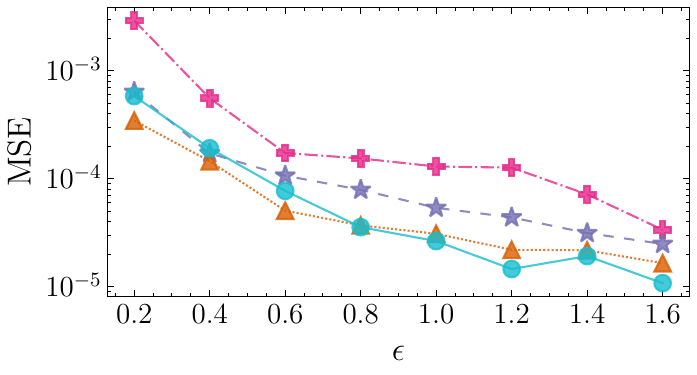}
	}
	\vspace{-0.05in} 
	\\
	\subfloat[Adult]{
		\includegraphics[width=0.24\textwidth]{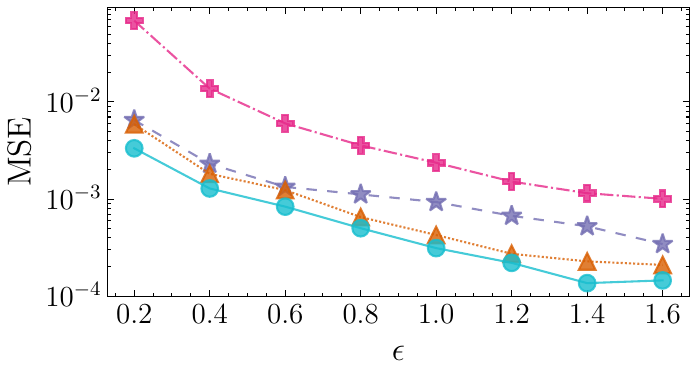}
	}
	\subfloat[Loan]{
		\includegraphics[width=0.24\textwidth]{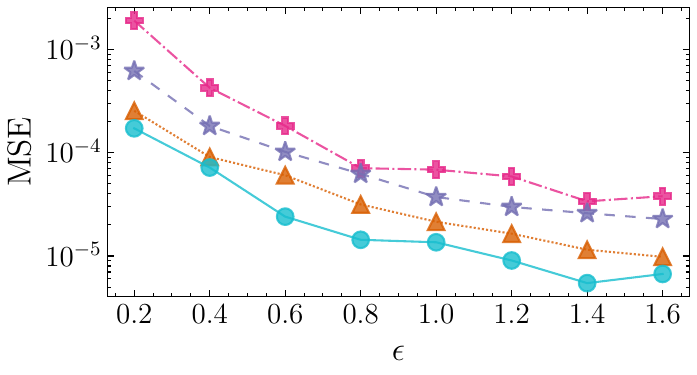}
	}
	\subfloat[Salary]{
		\includegraphics[width=0.24\textwidth]{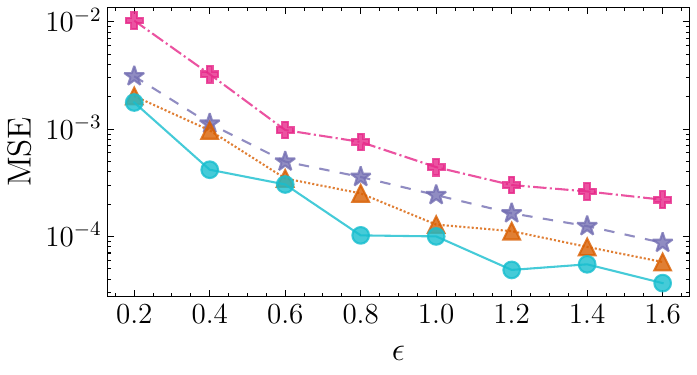}
	}
	\subfloat[Financial]{
		\includegraphics[width=0.24\textwidth]{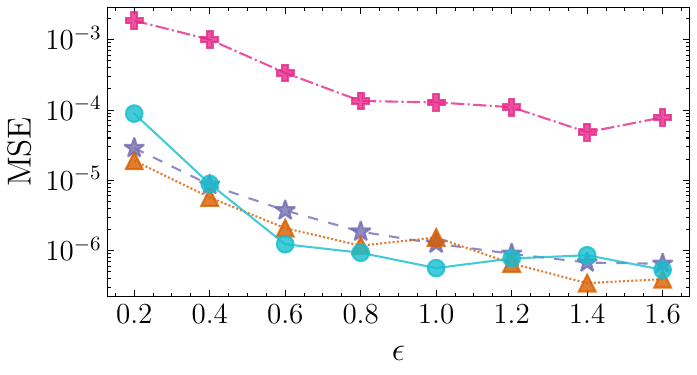}
	}
	\caption{Evaluation for 1-D Range Queries with Varying Privacy Budget $\epsilon$}
	\label{fig:evaluation_1d_epsilon}
\end{figure*}
\begin{figure*}
	\flushleft
	\includegraphics[width=0.35\textwidth]{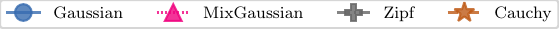}
	\vspace{-0.1in} 
	\hspace{6em}
	\includegraphics[width=0.35\textwidth]{figures/experiments/legend_1d.pdf}
	\\
	\centering
	\subfloat[User Allocatin Ratio $\alpha$]{
	\includegraphics[width=0.24\textwidth]{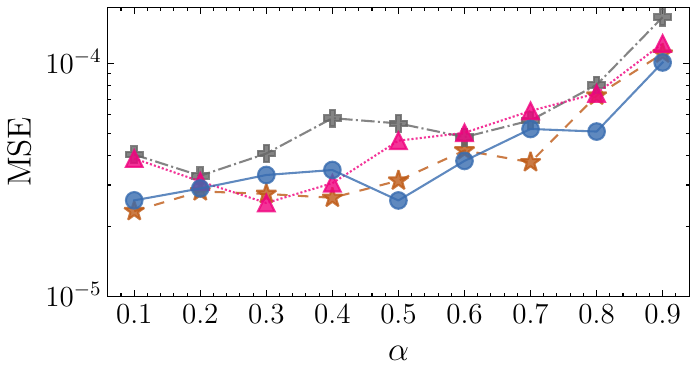}
	}
	\subfloat[Domain size $d$ on Gaussian]{
	\includegraphics[width=0.24\textwidth]{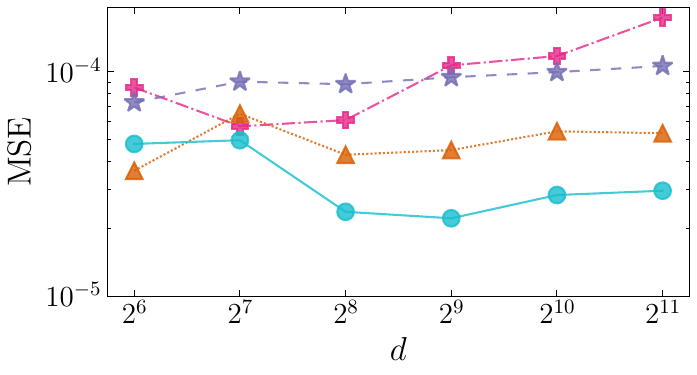}
	}
	\subfloat[Query Volume $Vol$ on Gaussian]{
	\includegraphics[width=0.24\textwidth]{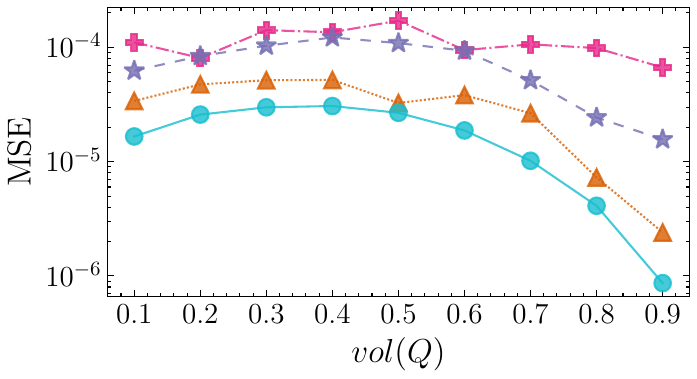}
	}
	\subfloat[User number $N$ on Gaussian]{
	\includegraphics[width=0.24\textwidth]{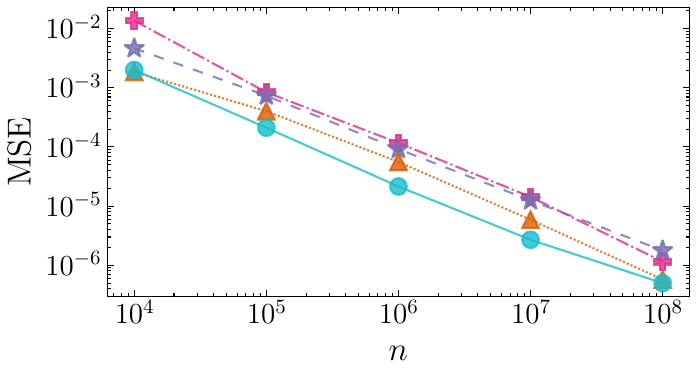}
	}
	\caption{Evaluation for 1-D Range Queries with Varying Parameters}
	\label{fig:evaluation_1d_params}
\end{figure*}

\subsection{Experimental Setting}

\textbf{Competitors.} We compare our methods with state-of-the-art techniques for range queries in LDP, including DHT \cite{cormode2019answering}, AHEAD \cite{du2021ahead}, PrivNUD \cite{wang2023privnud} for 1-D cases, and HDG \cite{yang2020answering}, AHEAD, PrivNUD, PRISM \cite{wang2022prism} for multi-D cases. We exclude hierarchical tree HH \cite{cormode2019answering} and HIO \cite{wang2019answering} baselines, as they have been demonstrated inferior to DHT and AHEAD in 1-D cases \cite{cormode2019answering, du2021ahead} and to HDG in multi-D cases \cite{yang2020answering}. For a fair comparison, we implement these methods using the codes and parameters from their original papers.

\textbf{Datasets.} We use four synthetic (Gaussian, MixGaussian, Cauchy, Zipf) and four real-world datasets (Adult \cite{adult}, Loan \cite{loan}, Salary \cite{salary}, and Financial \cite{financial}), each with 5 dimensions. On each dimension, the synthetic datasets Gaussian, Cauchy, and Zipf sample data from $Gaussian(0,1)$, $Cauchy(0,1)$, and $Zipf(1.1)$ distributions, respectively. The MixGaussian dataset follows a mixture of $Gaussian(0,0.5)$ and $Gaussian(3,0.8)$. For applying LDP mechanisms for frequency estimation and aligning with the existing works \cite{cormode2019answering, du2021ahead, yang2020answering, wang2022prism, wang2023privnud}, all datasets are bucketized into domain [1024] for 1-D evaluations and [256] for multi-D evaluations, except for some attributes in the real-world dataset with discrete values and original domain sizes less than our specified ones. We provide statistics of these datasets in Table~\ref{tab:datasets} where the mean and variance are for the default attribute in 1-D scenarios. A more detailed description is provided in Appendix~F of \cite{wang2024priplt}.

\textbf{Metrics.} We employ the mean square error (MSE) \cite{du2021ahead, wang2023privnud} to quantify the deviation between the estimated ($\tilde{f}_Q$) and actual ($f_Q$) answers to range queries $\mathbf{Q}$, denotes as $MSE(\mathbf{Q})=\sum_{Q\in \mathbf{Q}} (f_Q-\tilde{f}_Q)^2 / |\mathbf{Q}|$. During each evaluation, we test 1,000 randomly generated range queries with a specified query volume and report the final MSE by an average of 20 repeats of the experiment.

\textbf{Default Settings.} By default, we use a user allocation ratio $\alpha=0.2$ and an acceleration granularity factor $\phi=127$ for PriPL-Trees. For experiments, we set the privacy budget $\epsilon=0.8$ and the query volume $vol(Q) = 0.5$, where query volume represents the ratio of the query range size to the domain size on each attribute. All experiments use Python 3.11 on a Linux server with an Intel\circled{R} Xeon\circled{R} Gold 5218 CPU (2.3GHz) and 96GB of memory.

\subsection{1-D Experimental Results}
\label{subsec:1d_exp}

In this subsection, we evaluate the performance of PriPL-Tree and its competitors (DHT, AHEAD, and PrivNUD) for 1-D range queries and analyze the impact of data and query parameters on them.

\begin{table}[t]
\Small 
\setlength\tabcolsep{1.8pt} 
\caption{Summary of Datasets}
\centering
\begin{tabular}{|c|c|c|c|c|c|c|c|c|c|}
\hline
\textbf{Dataset} & \makecell[c]{\textbf{U.\#$^{\mathrm{a}}$}} & \makecell[c]{\textbf{L.\#$^{\mathrm{b}}$}} & \makecell[c]{\textbf{Mean}} & \makecell[c]{\textbf{Var.}} & \textbf{Dataset} & \makecell[c]{\textbf{U.\#}} & \makecell[c]{\textbf{L.\#}} & \makecell[c]{\textbf{Mean}} & \makecell[c]{\textbf{Var.}} \\ 
\hline
Gaussian & $10^6$ & 0  & 155.0 & 775.34 & Adult  & 32,561 & 4 & 30.36 & 336.88\\
\hline
MixGaussian & $10^6$ & 0 & 135.12 & 2516.62 & Loan & 148,045 & 3 & 48.67 & 1620.19\\ 
\hline
Cauchy & $10^6$ & 5 & 159.37 & 609.50 & Salary & 2,013,799 & 2 & 33.59 & 515.88\\ 
\hline
Zipf & $10^6$ & 5 & 24.40 & 2517.47 & Financial & 6,362,620 & 5 & 0.23 & 2.65\\ 
\hline
\multicolumn{10}{l}{$^{\mathrm{a}}$ U.\#: The number of users (i.e., samples).}\\
\multicolumn{10}{l}{$^{\mathrm{b}}$ L.\#: The number of leptokurtic attributes with a kurtosis exceeding 3.}
\end{tabular}
\label{tab:datasets}
\end{table}

\textbf{Overall Performance.}
We evaluate PriPL-Tree against three competitors across varying privacy budgets on synthetic and real-world datasets in Figure~\ref{fig:evaluation_1d_epsilon}. Our PriPL-Tree consistently outperforms competitors on most continuous distributions, such as Gaussian, MixGaussian, Cauchy, and those in the Adult, Loan, and Salary datasets. It significantly reduces MSEs by about 12.1\% to 66.6\%, averaging a 37.4\% reduction across different privacy settings. In highly leptokurtic distributions, like those in Zipf and Financial datasets, PriPL-Tree matches the performance of the leading competitor, PrivNUD. For these distributions, where a few values have significant frequencies, PriPL-Tree almost degenerates into an optimized hierarchical tree, similar to PrivNUD. It segregates high-frequency buckets into individual leaf nodes and merges low-frequencies into a single node, with both frequency and slope nearing zero.

\textbf{Impact of User Allocation Ratio $\alpha$.} In Figure~\ref{fig:evaluation_1d_params} (a), we assess the impact of the user allocation ratio $\alpha$ used in phase 1 of PriPL-Tree across four synthetic datasets. MSE remains stable for $\alpha\!\le\!0.5$ and slightly increases for $\alpha\!>\!0.5$, suggesting that fewer users are adequate for accurate PL fitting. Thus, we empirically set $\alpha\!=\!0.2$.

\textbf{Impact of Domain Size $d$.} In Figure~\ref{fig:evaluation_1d_params} (b), we explore the impact of domain size $d$ on the 1-D Gaussian dataset, demonstrating PriPL-Tree's superiority, particularly in large domains. Notably, PriPL-Tree performs less effectively in very small domains, where coarse bucketizing reduces the histogram's accuracy in representing distributions, resulting in suboptimal PL functions and inferior outcomes.

\textbf{Impact of Query Volume $vol(Q)$.} In Figure~\ref{fig:evaluation_1d_params} (c), we assess the impact of query volume on a Gaussian dataset. PriPL-Tree consistently records the lowest MSE. All methods display an MSE that increases initially and then decreases, peaking around $vol(Q)=0.5$. As detailed in Section~\ref{subsubsec:error_analysis}, the error for range queries correlates with the query range size and the frequency of intersections between query ranges and leaf nodes. Below $vol(Q)=0.5$, MSE increases primarily due to the expanding query range size. Above $vol(Q)=0.5$, MSE decreases as intersections occur more frequently at the domain's margins, where frequencies are lower and nearing 0, resulting in fewer PL fitting errors.

\textbf{Impact of User Number $N$.} In Figure~\ref{fig:evaluation_1d_params} (d), we examine the impact of user numbers with a Gaussian dataset. MSEs decrease as user numbers increase, aligning with the law of large numbers. However, PriPL-Tree's advantage diminishes with very small (e.g., $10^4$) or very large (e.g., $10^8$) user numbers. Insufficient users introduce excessive noise in LDP estimation, compromising PL parameter accuracy. Conversely, a larger user pool mitigates LDP noise, enabling even simple hierarchical trees to provide accurate estimates.

\textbf{Runtime Comparison:} In Figure~\ref{fig:runtime_1d}, we compare the runtime of our method with competitors across various datasets, with all methods implemented in Python for consistency. Our construction time is generally under half a minute. On average, our construction time of 27.2s is shorter than the average of our three competitors of 29.8s. This advantage is mainly due to the concise tree structure of PriPL-Tree, which utilizes fewer nodes. Additionally, our average query time is significantly lower, at around 50$\mu$s, while our competitors' times remain in the millisecond range.

\begin{figure}
	\centering
	\includegraphics[width=0.35\textwidth]{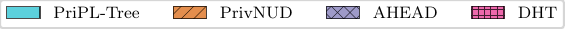}
	\vspace{-0.1in} 
	\\
	\subfloat[Construction Time]{
	\includegraphics[width=0.23\textwidth]{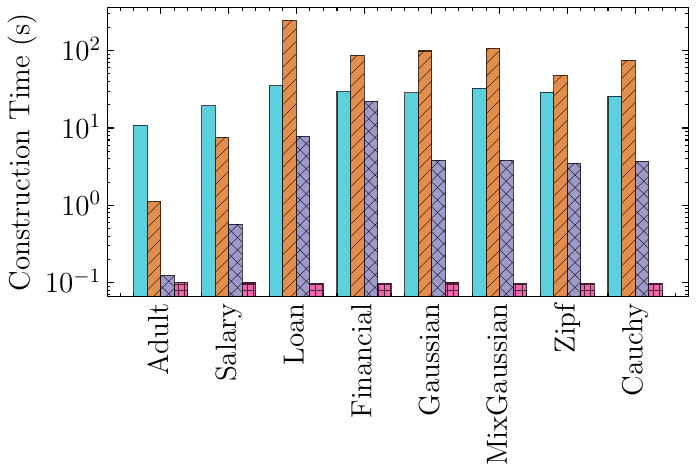}
	}
	\subfloat[Query Time]{
	\includegraphics[width=0.23\textwidth]{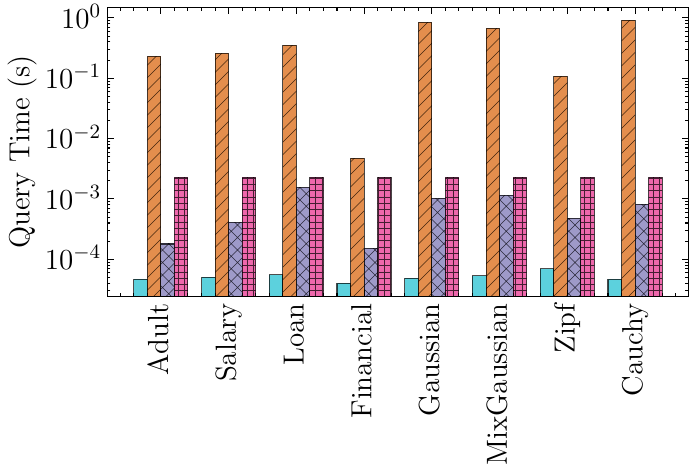}
	}
	\caption{Runtime Evaluation under Different Datasets}
	\label{fig:runtime_1d}
\end{figure}

\begin{figure*}
	\centering
	\includegraphics[width=0.45\textwidth]{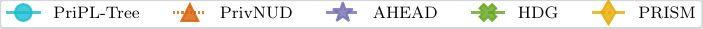}
	\vspace{-0.1in} 
	\\
	\subfloat[Gaussian]{
	\includegraphics[width=0.24\textwidth]{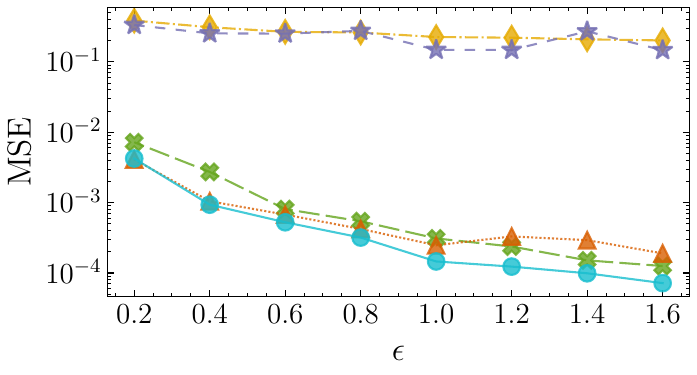}
	}
	\subfloat[MixGaussian]{
	\includegraphics[width=0.24\textwidth]{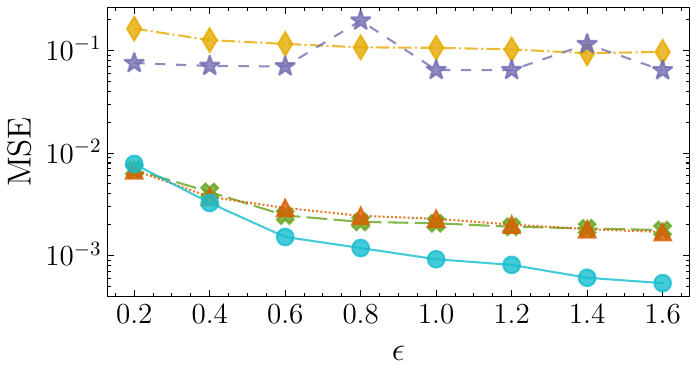}
	}
	\subfloat[Cauchy]{
	\includegraphics[width=0.24\textwidth]{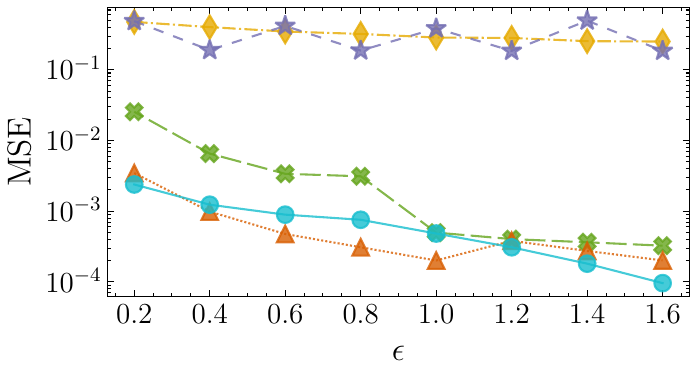}
	}
	\subfloat[Zipf]{
	\includegraphics[width=0.24\textwidth]{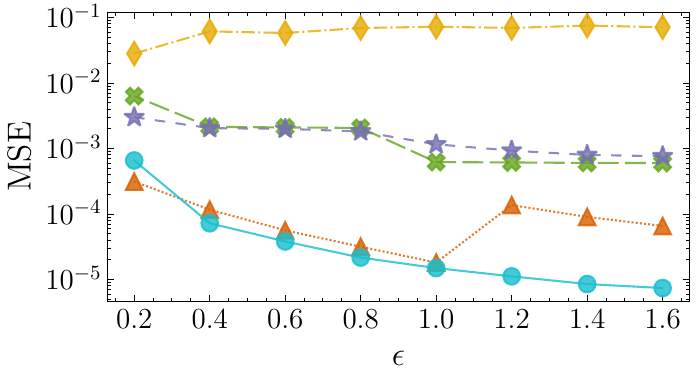}
	}
	\vspace{-0.05in} 
	\\
	\subfloat[Adult]{
	\includegraphics[width=0.24\textwidth]{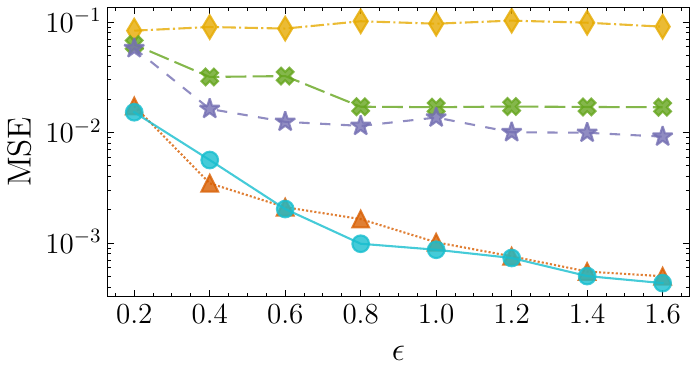}
	}
	\subfloat[Loan]{
	\includegraphics[width=0.24\textwidth]{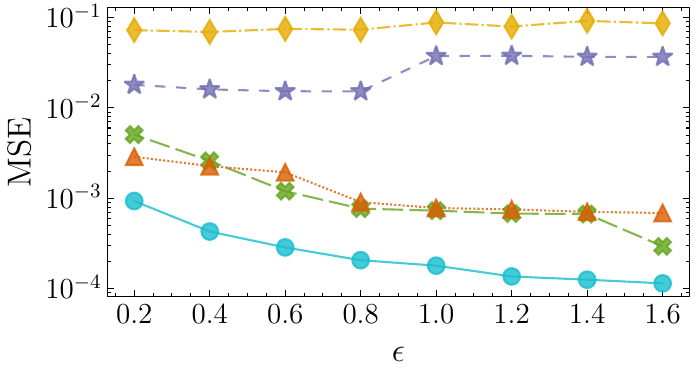}
	}
	\subfloat[Salary]{
	\includegraphics[width=0.24\textwidth]{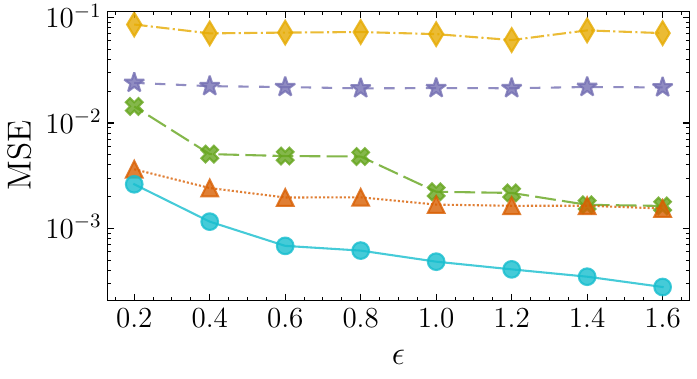}
	}
	\subfloat[Financial]{
	\includegraphics[width=0.24\textwidth]{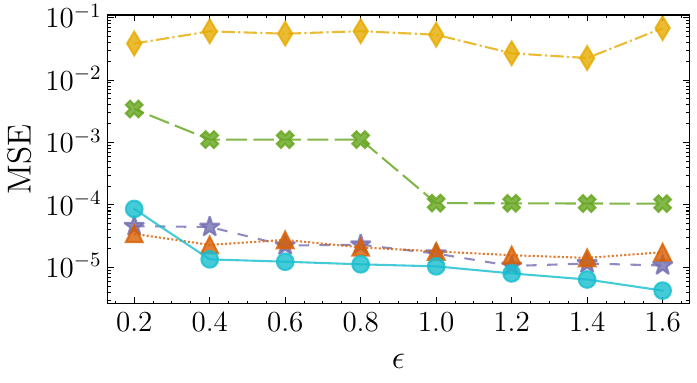}
	}
	\caption{Evaluation for 2-D Range Queries on 5-D Datasets with Varying Privacy Budget $\epsilon$}
	\label{fig:evaluation_md_epsilon}
\end{figure*}

\begin{figure*}
	\centering
	\includegraphics[width=0.45\textwidth]{figures/experiments/legend_md.pdf}
	\vspace{-0.1in} 
	\\
	\subfloat[Data dimension $m$ on Gaussian]{
	\includegraphics[width=0.24\textwidth]{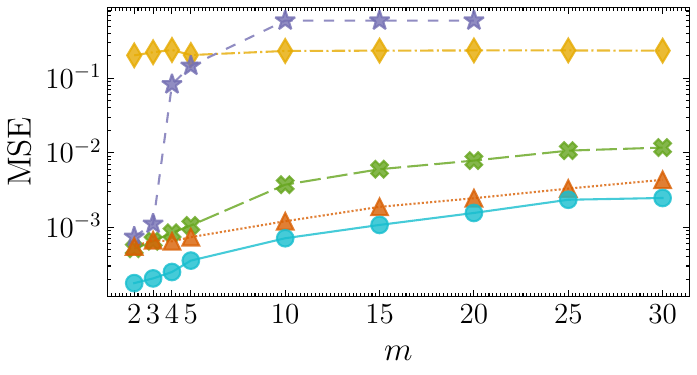}
	}
	\subfloat[Query Dimension $\lambda$ on Gaussian]{
	\includegraphics[width=0.24\textwidth]{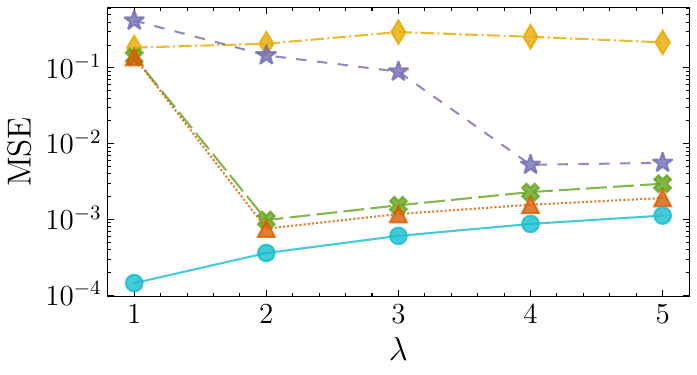}
	}
	\subfloat[Covariance $Cov$ on Gaussian]{
	\includegraphics[width=0.24\textwidth]{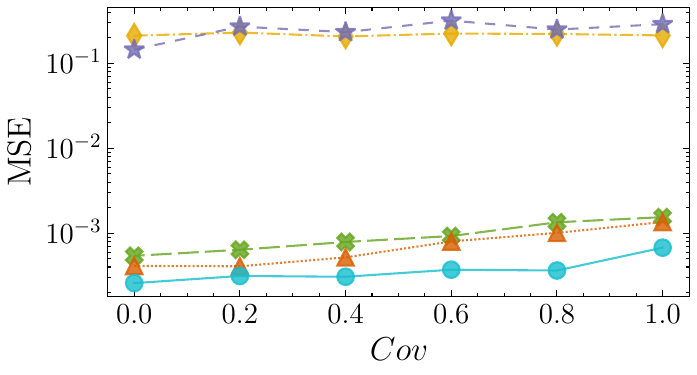}
	}
	\subfloat[Covariance $Cov$ on MixGaussian]{
	\includegraphics[width=0.24\textwidth]{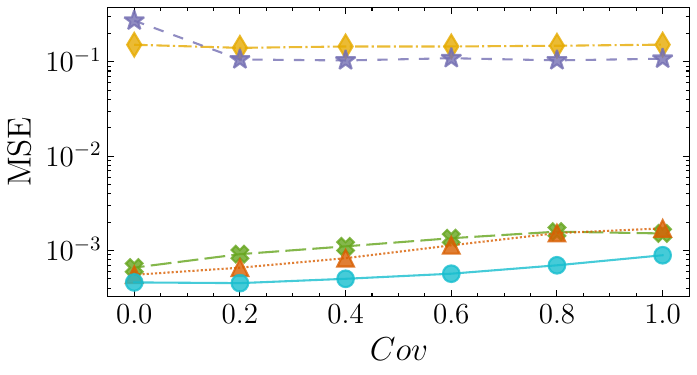}
	}
	\caption{Evaluation for Multi-D Range Queries}
	\label{fig:evaluation_md_params}
\end{figure*}

\subsection{Multi-D Experimental Results}
For multi-dimensional scenarios, we evaluate the performance of our method, PriPL-Tree with adaptive grids, against four competitors (HDG, AHEAD, PrivNUD, and PRISM). Typically, the standard experimental setup is evaluating 2-D queries on 5-D datasets, as all high-dimensional query results are derived from these 2-D queries.

\textbf{Overall Performance.} We evaluate the PriPL-Tree method against competitors under various privacy budgets on 5-D synthetic and real-world datasets, as shown in Figure~\ref{fig:evaluation_md_epsilon}. Utilizing adaptive grids, the PriPL-Tree method achieves the lowest MSEs in most cases, averaging 47.9\% lower MSE on real-world datasets and 23.7\% lower on synthetic datasets compared to state-of-the-art solutions. The extent of improvement of PriPL-Tree varies with the characteristics of different data distributions. In the Gaussian, MixGaussian, Loan, and Salary datasets, where most attributes are not leptokurtic, our PriPL-Tree method reduces MSE by 10.6\% to 81.9\%, averaging a reduction of 56.7\%. Conversely, for the other datasets, including Cauchy, Zipf, Adult, and Financial, which feature predominantly leptokurtic distributions where only a few values have significant frequencies, PriPL-Tree performs similarly to an optimized hierarchical tree, yielding a modest average MSE reduction of 14.9\%.

\textbf{Impact of Data Dimension $m$.} In Figure~\ref{fig:evaluation_md_params} (a), we evaluate PriPL-Tree and competitors across varying data dimensions on the Gaussian dataset with a covariance of 0.6. PriPL-Tree consistently shows the lowest MSEs across different dimensions. As expected, all MSEs increase with the dimension $m$ as users are distributed among more parts for estimation. Notably, two data points for AHEAD at $m \in \{25,30\}$ are missing in the figure due to exceeding our server's 96GB memory capacity, as per their open-source code. Despite these omissions, the available data points are sufficient to demonstrate that AHEAD's performance is inferior to ours.

\textbf{Impact of Query Dimension $\lambda$.} In Figure~\ref{fig:evaluation_md_params} (b), we assess PriPL-Tree and other methods across varying query dimensions on a Gaussian dataset with $Cov=0.6$. PriPL-Tree consistently records the lowest MSE, particularly noticeable in 1-D queries. During these experiments, we construct data structures across all five dimensions, with each 1-D range query selecting a dimension at random. This setup requires dividing users among $5 + \binom{5}{2} = 15$ parts, leading to fewer users per dimension and generally poorer 1-D estimations for competitors like HDG, AHEAD, PrivNUD, and PRISM. In contrast, our robust PriPL-Tree, enhanced by the consistency refinement in phase 3, effectively improves accuracy by updating frequencies and slopes in PriPL-Trees using all related 2-D adaptive grids.

\textbf{Impact of Attribute Correlation.} In Figures~\ref{fig:evaluation_md_params} (c) and (d), we examine the impact of attribute correlation on range queries using Gaussian and MixGaussian datasets. We use covariance to represent attribute correlation. The results reveal that PriPL-Tree consistently achieves the lowest MSEs, especially in datasets with high attribute correlations. This highlights our method's superior capability to capture underlying distributions with adaptive 2-D grids, unlike competitors that rely on uniform grids.

\section{Related Work}
\label{sec:related_work}

In this section, we review related works in both central differential privacy (DP) and LDP scenarios. 

\textbf{Range Query under DP:} 
In DP scenarios, the far-reaching hierarchical tree and constrained inference method were first proposed by Hay et al. \cite{hay2009boosting} and later optimized by Qardaji et al. \cite{qardaji2013differentially}. Various optimizations have since been proposed to mitigate noise errors on trees: Xiao et al. \cite{xiao2010differential} enhanced trees using Haar wavelet transforms; Cormode et al. \cite{cormode2012differentially} proposed a geometric privacy budget allocation method; Li et al. \cite{li2014data} optimized the non-uniform domain partitioning and privacy budget allocation based on data distribution and query workloads; Zhang et al. \cite{zhang2016privtree} proposed PrivTree (i.e., a Quad-tree ) with optimized node decomposition; Huang et al. \cite{huang2021approximate} employed a balanced box-decomposition tree (BBD-tree) for counting arbitrarily shaped geometric ranges. Beyond hierarchical trees, Qardaji et al. \cite{qardaji2013differentially} presented the grid method with optimized granularity, claiming grids are more suitable for high-dimensional queries. Recently, Zeighami et al. \cite{zeighami2021neural} introduced a model-driven approach that learns noisy answers from multiple 2-D range count queries to predict results without complex indexes.

\textbf{Range Query under LDP:}
In this context, we summarize existing methods according to tree-based and grid-based methods. 
In tree-based methods, HH \cite{cormode2019answering} and HIO \cite{wang2019answering} proposed the basic hierarchical tree in LDP almost simultaneously and optimize the tree's fan-out (i.e., branching factor). As improvements, AHEAD \cite{du2021ahead} merged intervals with low frequencies; PrivNUD \cite{wang2023privnud} customized the fan-out for each node; DHT \cite{cormode2019answering} optimized this tree via Haar wavelet transformation as in \cite{xiao2010differential}.
In grid-based methods, HDG \cite{yang2020answering} proposed the state-of-the-art hybrid dimensional grids, and PRISM \cite{wang2022prism} replaced simple grids with prefix-sum (PS) cubes. 
Additionally, there are other research topics covering range questions. \citet{mckenna2020workload} proposed a general matrix mechanism for linear queries in LDP, which can also be applied to answer range queries. And \citet{ye2021privkvm} explored PrivKVM* for range-based estimation in key-value datasets.

\textbf{Marginal Release under LDP:}
As a relevant problem to this work, we also review marginal release methods in LDP. Cormode et al. \cite{cormode2018marginal} introduced a Fourier transform-based method for private marginal release. Ren et al. \cite{ren2018textsf} explored multi-dimensional joint distribution estimation using the Expectation-Maximization (EM) algorithm and Lasso regression. A more advanced method is CALM, proposed by Zhang et al. \cite{zhang2018calm}, which extends the idea of PriView \cite{qardaji2014priview} from central DP to LDP and reconstructs high-dimensional marginals using low-dimensional estimations. This idea has been widely adopted in multi-dimensional range queries.

\section{Conclusion}
\label{sec:conclusion}

In this paper, we propose the PriPL-Tree to accurately answer range queries on arbitrary data distributions. The key idea is to approximate the underlying distribution using piecewise linear functions, which alleviates both non-uniform error and LDP noise error. We further extend this with adaptive grids to handle multi-dimensional cases, where the grids dynamically adjust to the data density, thus more accurately modeling the 2-D distribution and improving accuracy for multi-dimensional range queries. Extensive experiments on both real and synthetic datasets demonstrate the effectiveness and superiority of PriPL-Tree over state-of-the-art solutions. 

For future work, we will explore automatic and data-aware machine learning models to further enhance estimation in LDP scenarios.

\begin{acks}
 This work was supported by the National Natural Science Foundation of China under Grants 62172423, 92270123 and 62372122, and in part by the Research Grants Council, Hong Kong SAR, China under Grants 15209922, 15208923 and 15210023.  
\end{acks}

\bibliographystyle{ACM-Reference-Format}
\bibliography{PriPL-Tree.bib}

\clearpage

\appendix
\section{Technical Details for 1-D PriPL-Tree}
\label{appendix:1d_detail}

In this section, we detail the computation of Eq.(\ref{eq:sub_range}) for responding to 1-D range queries and provide examples to illustrate interval partitioning and its two included strategies.

\subsection{Response to 1-D Range Query}
\label{appendix:query}

In this subsection, we focus on computing $Q([l_{\text{sub}}, r_{\text{sub}}])$ within node $n_k$. Given the interval $I_k=[s_{k-1},s_k]$, slope $\tilde{\beta}_k$, and node frequency $\tilde{f}_k$ of node $n_k$, we first calculate the corresponding linear function $f= \tilde{\beta}_k\cdot v + b_k$, where $b_k$ is the intercept. This function is represented by the green line in Figure~\ref{fig:sub_query_examp}. Because $\tilde{f}_k = \sum_{v\in I_k} (\tilde{\beta}_k \cdot v + b_k)$, we deduce that $b_k = \frac{\tilde{f}_k}{|I_k|} - \frac{\tilde{\beta}_k\cdot(2s_{k-1} + |I_k| - 1)}{2}$. Next, we compute $Q([l_{\text{sub}}, r_{\text{sub}}])$, which represents the frequency of the red-hatched region in Figure~\ref{fig:sub_query_examp}:
\begin{align*}
	&Q([l_{\text{sub}}, r_{\text{sub}}]) \\
	&= \sum_{v\in [l_{sub}, r_{sub}]} (\tilde{\beta}_k \cdot v + b_k) \\
	&= \tilde{\beta}_k \cdot \frac{(l_{sub} + r_{sub})\cdot(r_{sub}-l_{sub}+1)}{2} + b_k \cdot(r_{sub}-l_{sub}+1) \\
	&= (r_{\text{sub}} - l_{\text{sub}} + 1)\cdot\left(\tilde{\beta}_k \left(\frac{l_{\text{sub}} + r_{\text{sub}} + 1 - |I_k|}{2} - s_{k-1}\right) + \frac{\tilde{f}_k}{|I_k|}\right).
\end{align*}

\begin{figure}[h]
\centering
\includegraphics[width=0.4\textwidth]{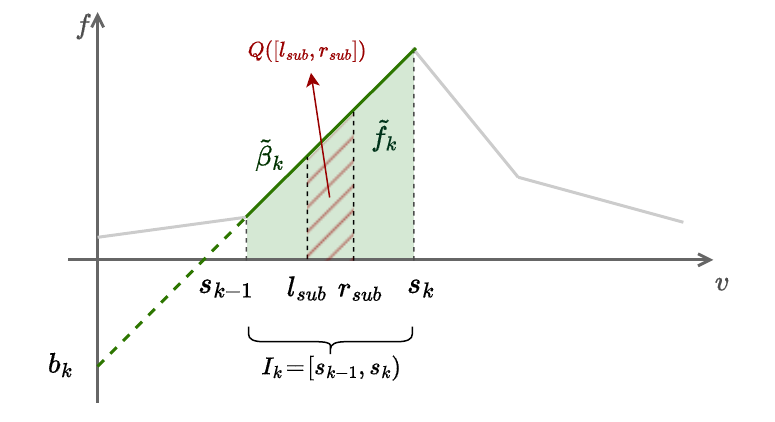}
\caption{An Example of $Q([l_{\text{sub}}, r_{\text{sub}}])$ within node $n_k$.}
\label{fig:sub_query_examp}
\end{figure}

\subsection{Interval Partitioning}

In this subsection, we illustrate the basic workflow of interval partitioning in Figure~\ref{fig:pl_fitting_example}~(a). During interval partitioning, we have presented two strategies — twice partitioning and search acceleration — to enhance effectiveness and efficiency, respectively, in Section~\ref{subsubsec:error_analysis}. Here, we provide additional examples of these two strategies for clarity and demonstrate their benefits.

\begin{figure}[h]
\centering
\includegraphics[width=0.48\textwidth]{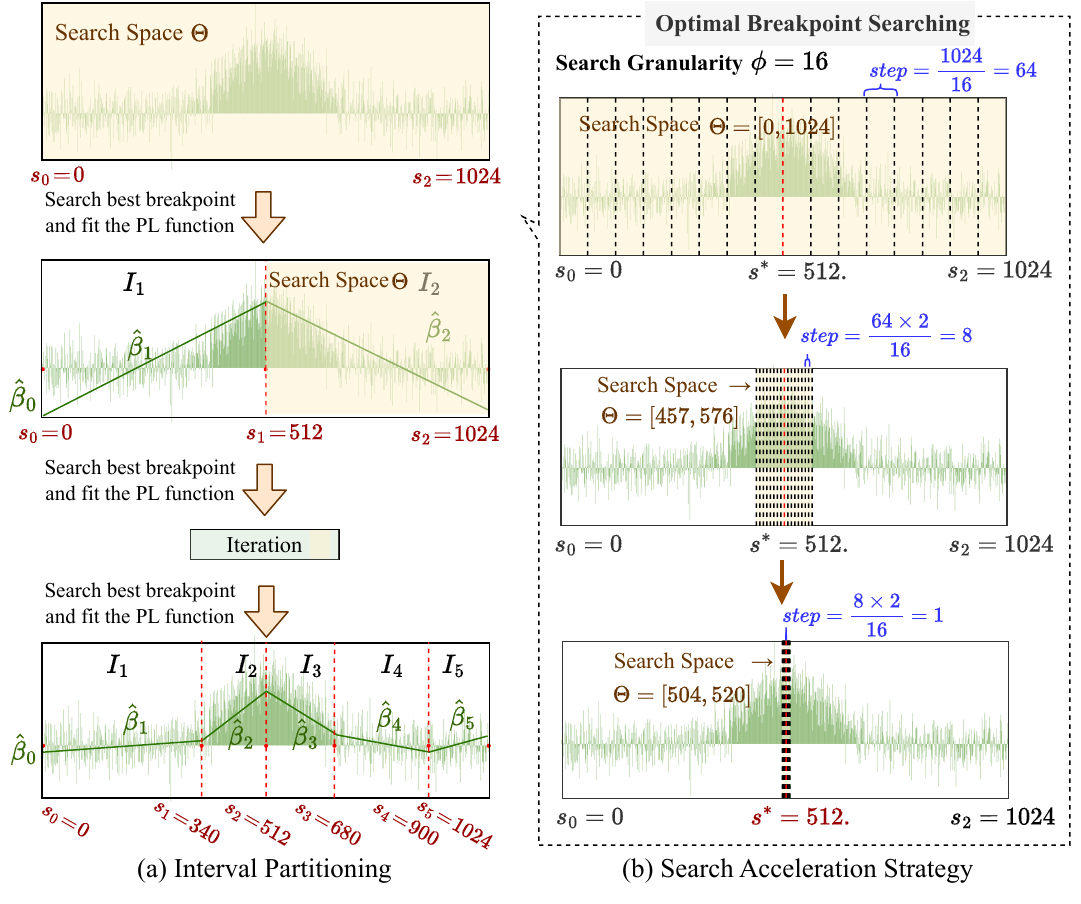}
\caption{An Example of Interval Partitioning.}
\label{fig:pl_fitting_example}
\end{figure}

\subsubsection{Twice Partitioning Strategy}

To illustrate the twice partitioning strategy, we present its results on the Cauchy and Salary datasets in Figure~\ref{fig:twice_partition}. The blue solid lines, representing $\hat{\mathbf{F}}^{EM}$ derived by EM in SW, are jagged and noisy, while the red solid lines, representing $\hat{\mathbf{F}}^{EMS}$ derived by EMS in SW, are smoother but tend to flatten peak features. Initial partitions based on $\hat{\mathbf{F}}^{\text{EM}}$ are marked by dashed blue lines and help identify critical peaks indicated by yellow stars. Subsequent partitions on $\hat{\mathbf{F}}^{\text{EMS}}$ are shown with red dashed lines, identifying critical turning points marked by black stars. In combination, this strategy facilitates the discovery of the most crucial partitions for precise PL fitting.

It is important to note that while $\hat{\mathbf{F}}^{\text{EMS}}$ does not directly provide accurate slopes at this phase, they can be refined during the PriPL-Tree refinement phase, deriving piecewise linear lines denoted by the black solid lines in Figure~\ref{fig:twice_partition}. Additionally, this strategy does not significantly increase the time complexity of interval partitioning. Partitioning over the jagged $\hat{\mathbf{F}}^{\text{EM}}$ typically converges quickly, and the partitioning over the smoother $\hat{\mathbf{F}}^{\text{EMS}}$, building on initial partitioning results, also converge swiftly. The total number of partitions remains within the maximum segment limit $K_{\max}$.

\begin{figure}[h]
	\centering
	\includegraphics[width=0.45\textwidth]{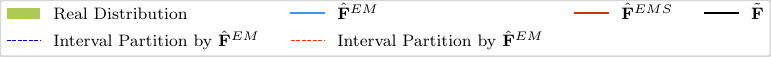}
	\vspace{-0.1in} 
	\\
	\subfloat[Cauchy]{
	\includegraphics[width=0.23\textwidth,valign=t]{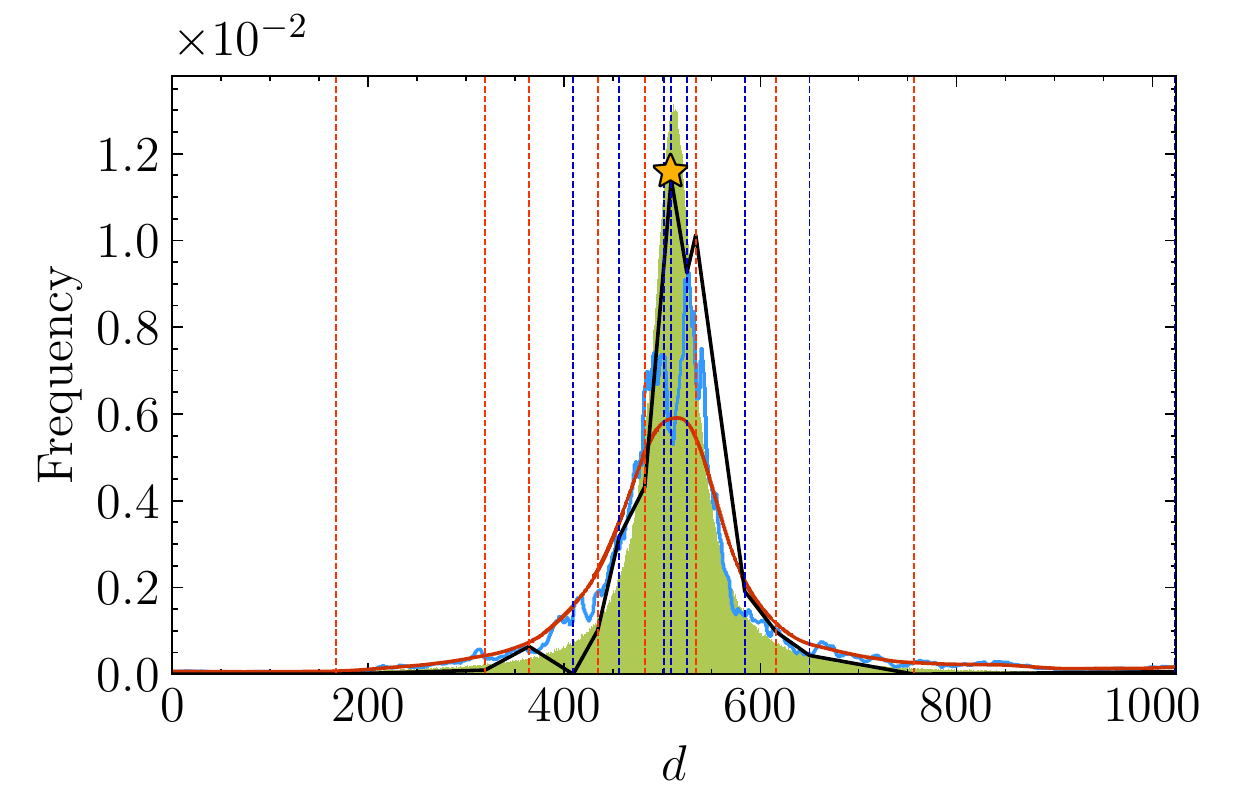}
	}
	\subfloat[Salary]{
	\includegraphics[width=0.23\textwidth,valign=t]{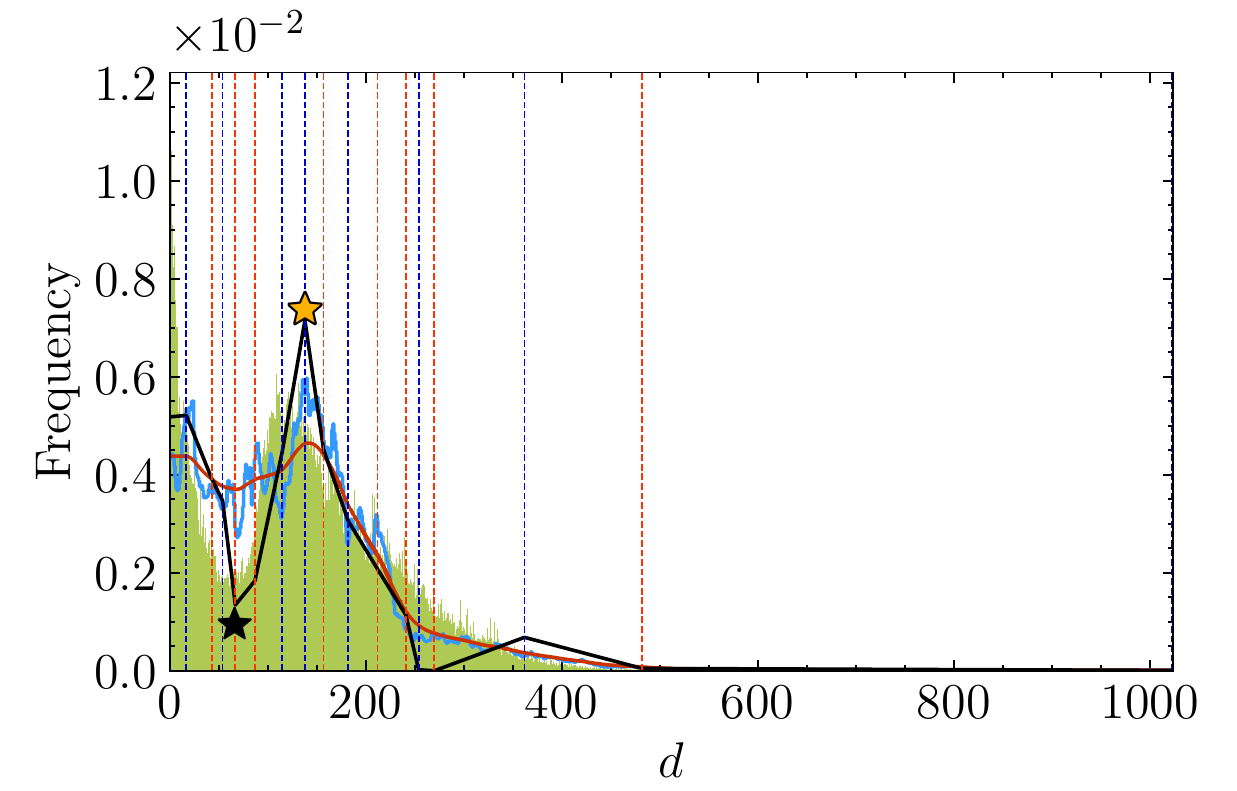}
	}
	\caption{An Example of Twice Partitioning with $\epsilon = 0.2$}
	\label{fig:twice_partition}
\end{figure}

\subsubsection{Search Acceleration Strategy}

In Figure~\ref{fig:pl_fitting_example}~(b), we illustrate the search acceleration strategy, which speeds up the initial breakpoint search during interval partitioning. Using a granularity factor of $\phi = 16$, we reduce the search step from 64 to 1, limiting the breakpoint candidates to no more than 16 per search. Black dashed lines mark the candidates, with the red dashed line indicating the optimal breakpoint. This method, compared to the basic one traversing all domain values, reduces the times of PL fitting for breakpoints from $O(d)$ to $O(\phi \cdot \log_{\phi} d)$. Although this acceleration might lead to a local rather than global optimal breakpoint, it has minimal impact on final results since each breakpoint in a piecewise linear function naturally represents a local optima to partition the domain. Moreover, because PL fitting occurs on a noisy histogram, a global optimal breakpoint might not accurately represent the actual distribution. Using a multi-granular search strategy that increases randomness can help mitigate this issue.

To validate the search acceleration strategy, we evaluated both the accelerated PriPL-Tree construction time and the MSE of queries on corresponding trees across four synthetic and four real-world datasets in Figure~\ref{fig:evaluation_seach_granularity}. A granularity factor $\phi$ of $1024$, meaning all values in the domain are traversed, represents no acceleration. The green bar indicates the accelerated private PL fitting time, which increases with $\phi$. The red line shows the MSE, fluctuating irregularly with increases in $\phi$, yet the variation between the highest and lowest MSEs remained within a two-fold difference. Observationally, a larger $\phi$ in the range of $[128,256]$ often provides better utility while effectively speeding up the process.

\begin{figure}
	\centering
	\includegraphics[width=0.45\textwidth]{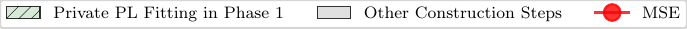}
	\vspace{-0.1in} 
	\\
	\subfloat[Gaussian]{
	\includegraphics[width=0.24\textwidth]{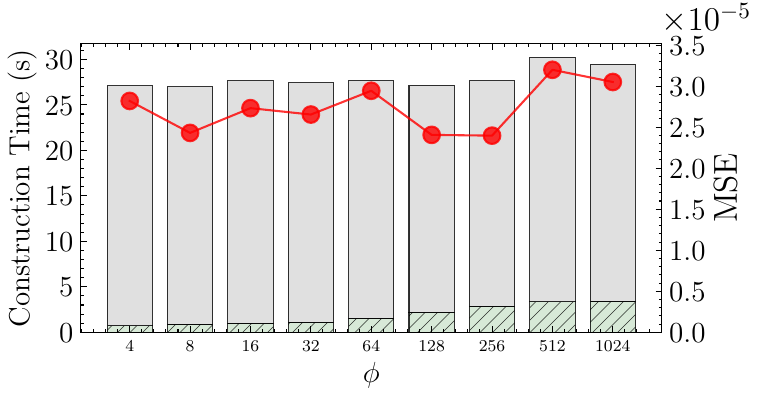}
	}
	\subfloat[MixGaussian]{
	\includegraphics[width=0.24\textwidth]{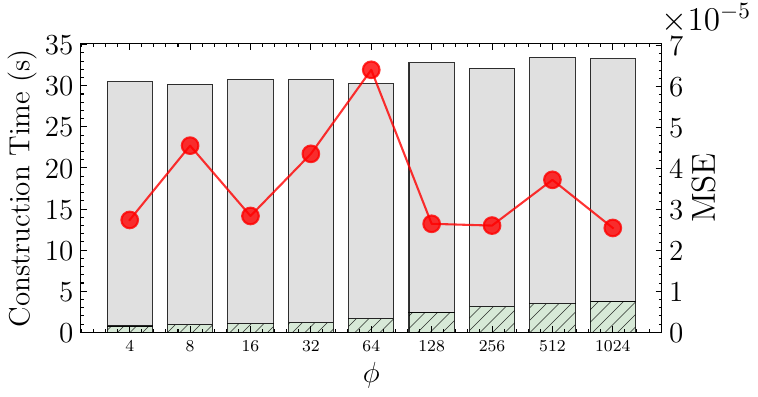}
	}
	\\
	\subfloat[Cauchy]{
	\includegraphics[width=0.24\textwidth]{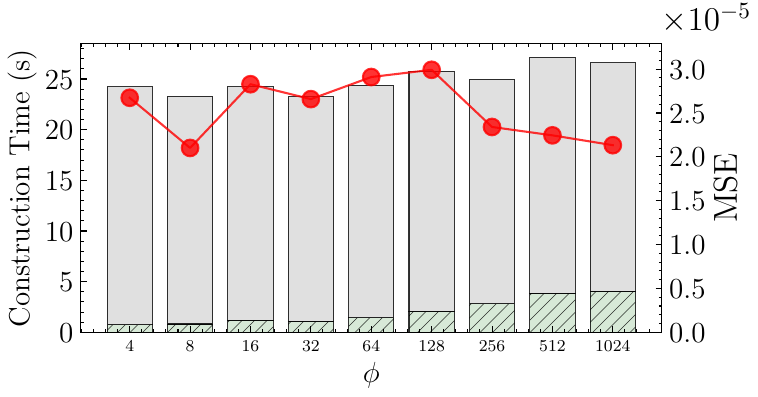}
	}
	\subfloat[Zipf]{
	\includegraphics[width=0.24\textwidth]{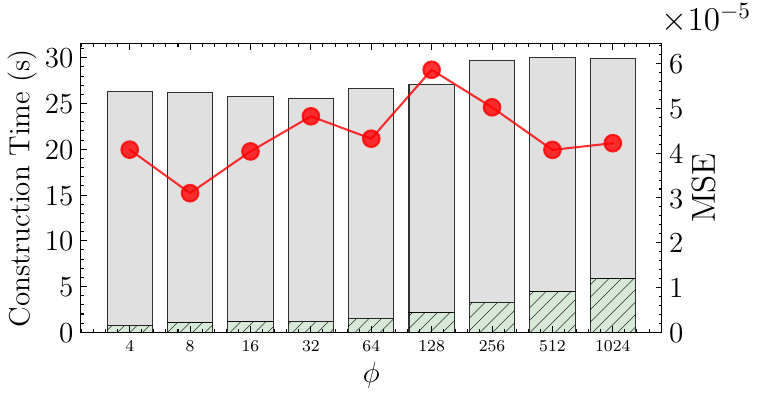}
	}
	\\
	\subfloat[Adult]{
	\includegraphics[width=0.24\textwidth]{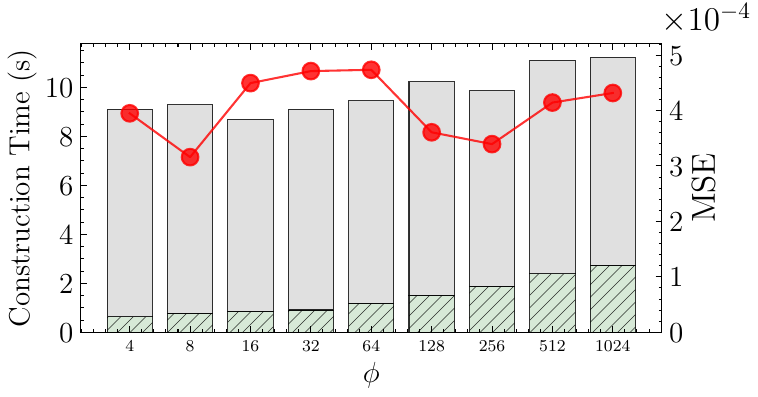}
	}
	\subfloat[Loan]{
	\includegraphics[width=0.24\textwidth]{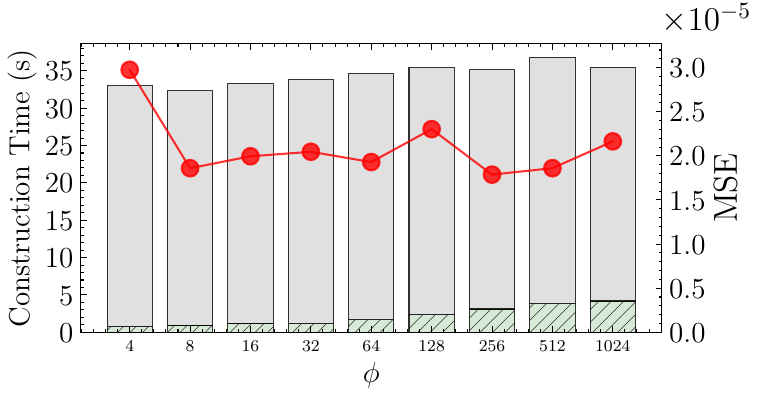}
	}
	\\
	\subfloat[Salary]{
	\includegraphics[width=0.24\textwidth]{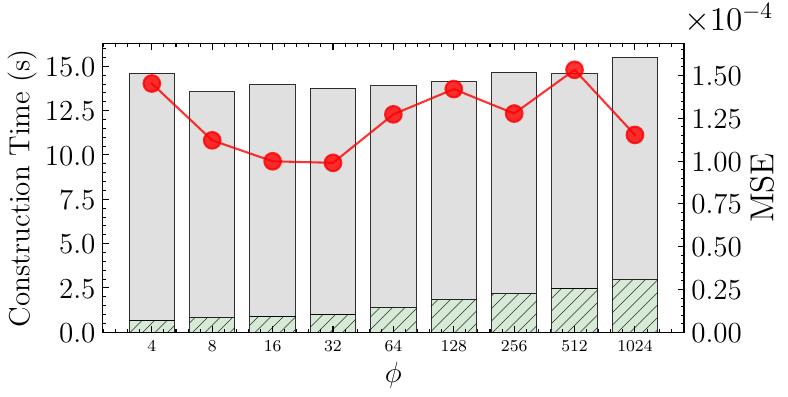}
	}
	\subfloat[Financial]{
	\includegraphics[width=0.24\textwidth]{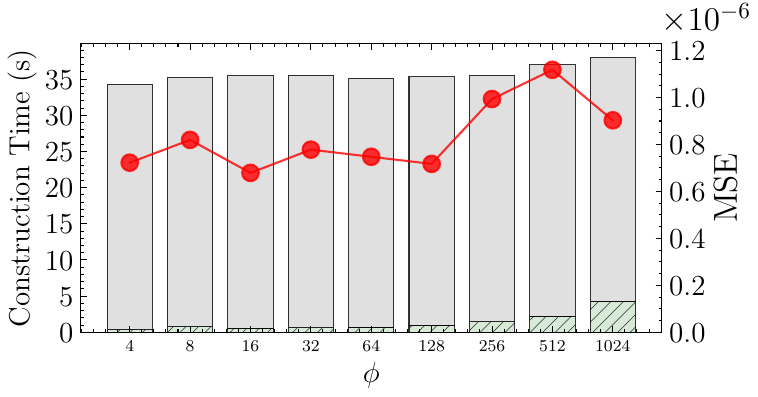}
	}
	\caption{Construction Time and MSE for PriPL-Tree with Search Acceleration Strategy across Different Granularity Factors $\phi$}
	\label{fig:evaluation_seach_granularity}
\end{figure}

\section{Technical Details for Multi-Dimensional Range Queries}

In this section, we illustrate the entire workflow and detail the algorithm for adaptive 2-D grid partitioning.

\subsection{The Workfolow}
\label{appendix:md_detail}

We provide an example of the workflow in Figure~\ref{fig:adaptive_grids_workflow} to illustrate this process further.

\begin{figure*}[htbp]
\centering
\includegraphics[width=\textwidth]{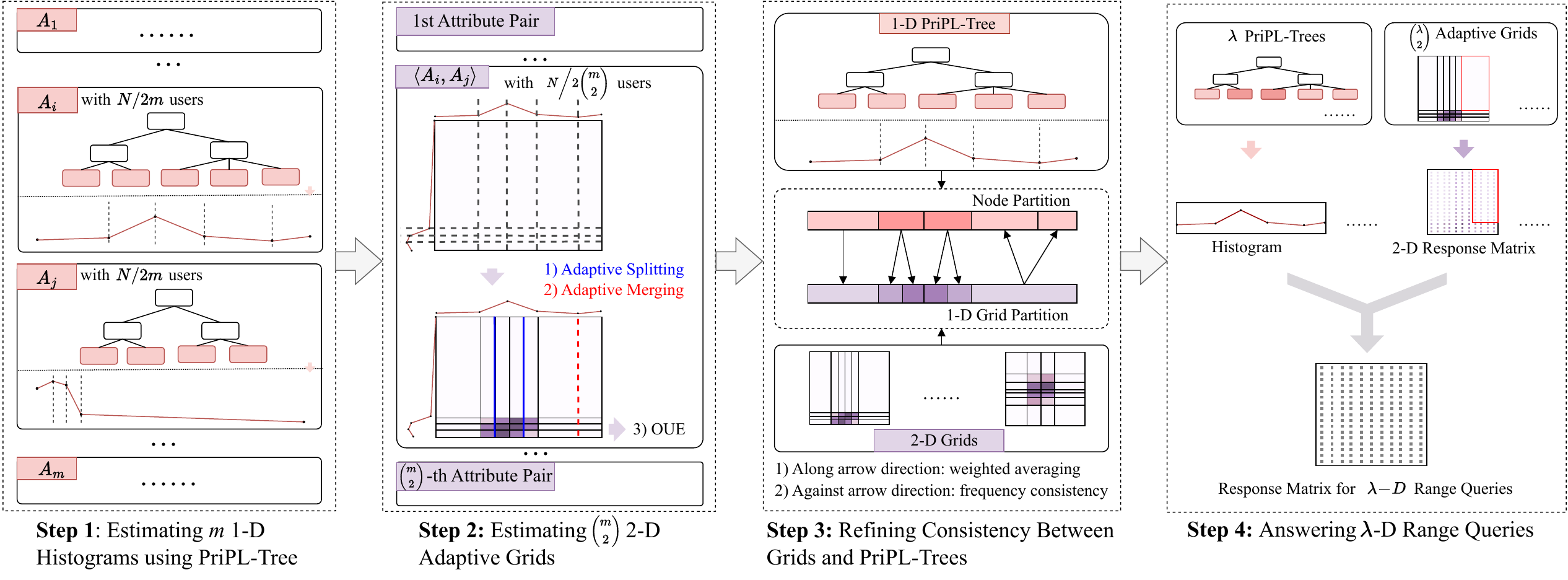}
\caption{Workflow of Multi-dimensional Range Queries}
\label{fig:adaptive_grids_workflow}
\end{figure*}

\subsection{The Algorithm For Adaptive Grid Partitioning}
\label{appendix:adaptive_partition}

Building on the concept of adaptive partitioning described in Section~\ref{subsec:adap_grid}, we detail the algorithm in Algorithm~\ref{algo:grid_partition} and provide an example in step 2 of Figure~\ref{fig:adaptive_grids}. For the attribute pair $\langle A_i, A_j \rangle$, we initially partition its domain based on the leaf node partitions in the PriPL-Trees, creating a grid $G$ with $g_i \times g_j$ cells (as shown in line 1). We then dynamically adjust partition lines along each dimension for adaptive partitioning, where adding a line splits a marginal cell (as shown in line 5) and removing a line merges marginal cells (as shown in line 13). In the algorithm, let $I$ index marginal cells on both $A_i$ and $A_j$; we consistently select the most significant cells for splitting (high-frequency cells) or merging (low-frequency cells). This adjustment of marginal cells adheres to two principles: ensuring each cell's frequency exceeds the standard deviation of the OUE noise $\bar{\sigma} = \sqrt{2\sigma^2 \cdot \binom{m}{2}}$, and minimizing the squared error $Err_G$ for range queries, as dictated by the splitting and merging conditions in lines 6 and 14.

\begin{algorithm}[tbp]
	\small
	\caption{Adaptive 2-D Grid Partitioning}
	\label{algo:grid_partition}
	\KwIn{\textls[-5]{Attribute pair $\langle A_i, A_j \rangle$ and their marginal histograms}}
	\KwOut{Grid $G$ with adaptive partitions}
	Initialize a Grid $G$ with size $g_i\times g_j$\;
	Set $S\!=\!\{I^i_1,I^i_2,\ldots,I^i_{g_i}\} \cup \{I^j_1,I^j_2,\ldots,I^j_{g_j}\}$ and $M\!=\!\phi$\;
	\While{$S\neq \phi$}{
		Select marginal cell $c_I$ ($I\!\in\!S$) with maximal frequency $\tilde{f}_{I}$\;
		Consider splitting $c_I$ into $c_{I_1}$ and $c_{I_2}$ with $\tilde{f}_{c_{I_1}}\!\approx\!\tilde{f}_{c_{I_2}}$\;
		\If{$\tilde{f}_{I}\!>\!\bar{\sigma}$ and $Err_G(S)\!>\!Err_G(S\!-\!\{I\}\!\cup\!\{I_1,\!I_2\})$}{
			Split $c_I$ and update $S \leftarrow S\!-\!\{I\} \cup\{I_1, I_2\}$\;
		}
		\Else{
			Delete $I$ from $S$ and add it to $M$\;
		}
	}
	\While{$M\neq \phi$}{
		Select marginal cell $c_I$ ($I\!\in\!M$) with minimal frequency\;
		\If{the index $I'$ of its neighbor exists in $M$}{
			Consider merging $c_I$ and $c_{I'}$ into $c_{I_3}$\;
			\If{$\tilde{f}_{I_3}\!\le\!\bar{\sigma}$ or $Err_G(M)\!>\!Err_G(M\!-\!\{I,I'\}\!\cup\!\{I_3\})$}{
				Merging $c_I$ and $c_{I'}$, update $M\!\leftarrow\!M\!-\!\{I,I'\}\!\cup\!\{I_3\}$\;
			}
			Delete $I'$ from $M$\;
		}
		Delete $I$ from $M$\;
	}
\end{algorithm}

\section{Error Analysis on 1-D PriPL-Tree}
\label{appendix:error_ana_1d}

In this section, we prove the asymptotic bound for noise and sampling error in PriPL-Tree and present a numerical method to accurately compute this error.

\subsection{Proof and Analysis of Theorem~\ref{theorem:pripl_tree_node_error}}

For convience, we restate the theorem as follows.

\begin{reptheorem}{theorem:pripl_tree_node_error}
	Given a PriPL-Tree with at most $K$ segments (corresponding to $K$ leaf nodes), the error variance of frequencies after weight averaging in refinement (phase~3) is $O\left(\frac{K\cdot \log K}{(1-\alpha) \cdot (K+1) \cdot N \cdot \epsilon^2}\right)$ for non-leaf nodes and $O\left(\frac{\log K}{(1-\alpha) \cdot N \cdot \epsilon^2}\right)$ for leaf nodes. After frequency consistency, these variance is capped at $O\left(\frac{K\log K}{(1-\alpha) N \epsilon^2}\right)$.
\end{reptheorem}

\subsubsection{Analysis}

It is important to note that after the frequency consistency step, some nodes may experience an increase in variance from $O\left(\frac{\log K}{(1-\alpha)N\epsilon^2}\right)$ to $O\left(\frac{K\log K}{(1-\alpha)N\epsilon^2}\right)$. However, most nodes exhibit reduced variances, which can be verified by computing the precise variances numerically. This phenomenon is common to any tree structure with non-uniform branchings, such as PrivNUD \cite{wang2023privnud} and AHEAD \cite{du2021ahead}, which exhibit heteroscedastic frequencies across nodes.
In particular, in cases where the tree features uniform branching at each layer and all leaf nodes are at the same depth, each layer's nodes will be allocated an equal number of users, resulting in uniform variance among their estimated frequencies. According to the Gauss-Markov theorem \cite{hay2009boosting}, the variance post-frequency consistency in such scenarios can be estimated as $O\left(\frac{\log K}{(1-\alpha) N \epsilon^2}\right)$.

\subsubsection{Proof} We prove the Theorem~\ref{theorem:pripl_tree_node_error} as follows.

\begin{proof}
For convenience, we assume the PriPL-Tree has a maximum height of $h_{\max} \le \log_2 K$. Each node utilizes $\alpha_k$ users, with $\alpha_k \ge \frac{1-\alpha}{h_{\max}} \ge \frac{1-\alpha}{\log_2 K}$, and each non-leaf node has $b_k$ branches with $1\le b_k \le K$. 

First, we prove the error bound after the weight averaging step. 

For leaf nodes $n_k$, there are two frequencies: $\hat{f}_k$ derived from the SW mechanism using $\alpha N$ users and $\bar{f}_k$ from the OUE mechanism using $\alpha_k$ users. By weighted averaging, the variance of the updated frequency $\dot{f}_k$ is given by
\begin{align*}
	\var(\dot{f}_k) 
	&= \frac{\var(\hat{f}_k)\cdot \var(\bar{f}_k)}{\var(\hat{f}_k) + \var(\bar{f}_k)} \\
	&=  \var(\bar{f}_k) - \frac{\var ^2(\bar{f}_k)}{\var(\hat{f}_k) + \var(\bar{f}_k)} \\
	&\le \var(\bar{f}_k) \\
	&= \frac{4e^{\epsilon}}{N \cdot \alpha_k \cdot(e^\epsilon-1)^2} \\
	&= O\left(\frac{1}{N\alpha_k \epsilon^2}\right)\\
	&= O\left(\frac{\log K}{(1-\alpha) N \epsilon^2}\right).
\end{align*}

For the non-leaf node $n_{p}$ with $b_p$ child nodes, its frequency is updated based on its own and its children's frequencies $\{\dot{f}_c | n_c \in \text{child}(n_p)\}$. Before computing the updated variance, we present two preliminary conclusions:

(1) For a child node $n_c$, its updated frequency variance $\var(\dot{f}_c)$ is less than or equal to the variance of the parent node's original frequency, i.e., $\var(\bar{f}_{p})$. Assuming $N'$ users are allocated to the subtree rooted at $n_{p}$ and $h_{p}$ denotes the subtree's maximum height. Since node $n_{p}$ estimates its frequency $\bar{f}_{p}$ with at most $N'/h_{p}$ users and its children are allocated at least $N'/h_{p}$ users, the child node's frequency variance $\var(\bar{f}_c)$ does not exceed that of $n_{p}$, i.e., $\var(\bar{f}_c) \le \var(\bar{f}_{p})$. Due to weighted averaging, the updated variance $\var(\dot{f}_c) \le \var(\bar{f}_c)$, implying $\var(\dot{f}_c) \le \var(\bar{f}_{p})$.

(2) Estimates for different child nodes $\{\dot{f}_c | n_c \in \text{child}(n_p)\}$ of node $n_p$ are independent, as each aggregates frequencies from its respective subtree, with no overlap among these subtrees.

As such, the variance of $n_{p}$ for non-leaf nodes can be deduced as follows:
\begin{align*}
	\var(\dot{f}_{p})
	&= \frac{\var(\bar{f}_{p}) \cdot \var\left(\sum_{n_c\in \text{child}(n_{p})} \dot{f}_c\right)}{\var(\bar{f}_{p}) + \var\left(\sum_{n_c\in \text{child}(n_{p})} \dot{f}_c\right)} \\
	&= \frac{\var(\bar{f}_{p}) \cdot \left(\sum_{n_c\in \text{child}(n_{p})} \var(\dot{f}_c)\right)}{\var(\bar{f}_{p}) + \sum_{n_c\in \text{child}(n_{p})} \var(\dot{f}_c)} \\
	&= \var(\bar{f}_{p}) - \frac{\var^2(\bar{f}_{p})}{\var(\bar{f}_{p}) + \sum_{n_c\in \text{child}(n_{p})} \var(\dot{f}_c)} \\
	&\le \var(\bar{f}_{p}) - \frac{\var^2(\bar{f}_{p})}{\var(\bar{f}_{p}) + b \cdot \var(\bar{f}_{p})} \\
	&= \frac{b_k}{b_k+1} \cdot \var(\bar{f}_{p}) \\
	&\le \frac{K}{K+1} \cdot \var(\bar{f}_{p}) \\
	&= O\left(\frac{K}{K+1} \cdot \frac{\log K}{(1-\alpha)N\epsilon^2}\right).
\end{align*}

Next, we analyze the error variance following the frequency consistency step. Let $n_c$ represent the child node pending update, and $n_p$ its parent. At this stage, the nodes $n_c \in \text{child}(n_p)$ are divided into two groups, $D_0$ and $D_+$. Nodes in $D_0$ typically display frequencies $\dot{f}_c$ close to or below zero, which are then adjusted to zero. This modification might slightly increase or decrease their variances but does not alter their variance bound. For nodes in $D_+$, frequency updates are based on both parental and sibling frequencies, leading to a variance as follows.
\begin{align*}
 	&\var(\tilde{f}_c) \\
 	&= \var\left(\frac{|D_+|-1}{|D_+|}\cdot \dot{f}_c + \frac{1}{|D_+|}\cdot \tilde{f}_p - \frac{1}{|D_+|} \sum_{n_s \in \text{child}(n_p)\slash n_c} \dot{f}_s\right)\\
 	&= \frac{(|D_+|-1)^2}{|D_+|^2} \cdot \var(\dot{f}_c) + \frac{1}{|D_+|^2} \cdot \var(\tilde{f}_p) \\
 	&\quad + \frac{1}{|D_+|^2} \cdot \sum_{n_s \in \text{child}(n_p)\backslash n_c}\var(\dot{f}_s) + \frac{2(|D_+|-1)}{|D_+|^2} \cdot \cov(\tilde{f}_p, \dot{f}_c) \\
 	&\quad - \frac{2}{|D_+|^2} \cdot \cov\left(\tilde{f}_p, \sum_{n_s \in \text{child}(n_p)\backslash n_c} \dot{f}_s \right) \\
 	&\le \frac{(|D_+|-1)^2}{|D_+|^2} \cdot \var(\dot{f}_c) + \frac{1}{|D_+|^2} \cdot \var(\tilde{f}_p) \\
 	&\quad + \frac{1}{|D_+|^2} \cdot \sum_{n_s \in \text{child}(n_p)\backslash n_c}\var(\dot{f}_s) \\
 	&\quad + \frac{2(|D_+|-1)}{|D_+|^2} \cdot \sqrt{\var(\tilde{f}_p) \cdot \var(\dot{f}_c)} \\
 	&\le \frac{|D_+|^2 + |D_+| - 1}{|D_+|^2} \max\left(\{\var(\tilde{f}_p)\} \cup \{\var(\dot{f}_j) | j \in \text{child}(n_p)\}\right) \\
 	&\le \frac{|D_+| + 1}{|D_+|} \max\left(\{\var(\tilde{f}_p)\} \cup \{\var(\dot{f}_j) | j \in \text{child}(n_p)\}\right) \\
 	&\le \left(\frac{|D_+| + 1}{|D_+|}\right)^{h_{\max}} \max\left(\{\var(\bar{f}_p)\} \cup \{\var(\bar{f}_j) | j \in \text{child}(n_p)\}\right) \\
 	&\le 2^{\log_2 K} \cdot O\left(\frac{\log K}{(1-\alpha) N \epsilon^2}\right) \\
 	&= O\left(\frac{K\log K}{(1-\alpha) N \epsilon^2}\right).
 \end{align*}
 \end{proof}

\subsection{The Numerical Method for Calculating Noise and Sampling Error}

The variances of node $n_k$'s frequency $\tilde{f}_k$ after PriPL-Tree Refinement can be veiwed as a weighted average of frequency estimates of all nodes from the first two phases \cite{qardaji2013understanding}, as we exemplified in Figure~\ref{fig:numerical_error}. Here, the vector $\mathbf{V}$ stores the original independent variances of each node, and each node $n_k$ maintains a weight vector $\mathbf{W}_k$ corresponding to each value in $\mathbf{V}$. Let $w_{k,j}$ denote the $j$-th value in $\mathbf{W}_k$ and $v_j$ denote the $j$-th value in $\mathbf{V}$. For node $n_k$, its final variance can be computed as $\sum_{j\le|\mathbf{V}|}w_{k,j}^2 \cdot v_j$. 

\begin{figure}[h]
\centering
\includegraphics[width=0.48\textwidth]{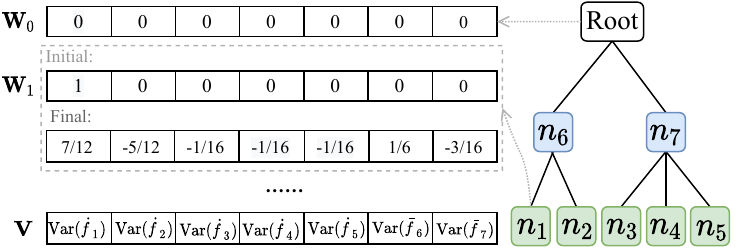}
\caption{An Example of Variance Computation.}
\raggedright
(For convenience, assuming all values in $\mathbf{V}$ are equal to $v^*$ and all frequencies remain positive during the frequency consistency step, we can derive weight $\mathbf{W}_1$ as illustrated in the figure. Consequently, the final updated variance for $\tilde{f}_1$ is $\var(\tilde{f}_1) = \frac{113}{192}v^*$.)
\label{fig:numerical_error}
\end{figure}

Following this idea, the main challenges in calculating the variance of each node's frequency involve computing the basic variance vector $\mathbf{V}$ and each node's weight vector $\mathbf{W}_k$. We provide an algorithm in Algorithm~\ref{algo:numerical_error} and present it as follows.

For the vector $\mathbf{V}$, each non-leaf node $n_k$ records the variance $\var(\bar{f}_k)$ (i.e., OUE's variance) within it, as shown in lines 6$\sim$7 in Algorithm~\ref{algo:numerical_error}. For each leaf node $n_k$, it stores the variance of the weighted updated frequency $\dot{f}_k$, which is computed by combining $\var(\bar{f}_k)$ derived according to OUE and $\var(\hat{f}_k)$ derived according to SW, as shown in line 5 in Algorithm~\ref{algo:numerical_error}. Unlike OUE, which yields an unbiased estimate with constant variance, SW leads to estimates that are biased and affected by the unknown original data distribution, making it difficult to express the variance directly \cite{li2020estimating}. As an alternative, we can compute its empirical error $(\hat{f}_k-\bar{f}_k)^2$ if we approximate $f_k$ by $\bar{f}_k$. Considering $\E((\hat{f}_k-\bar{f}_k)^2) = \E(((\hat{f}_k - f_k) + (\bar{f}_k-f_k))^2) = \E(\hat{f}_k - f_k)^2 + \var(\bar{f}_k)$ and that SW typically derives a better estimate with smaller variance than OUE \cite{li2020estimating}, we can use $(\hat{f}_k-\bar{f}_k)^2 - \var(\bar{f}_k)$ to approximate its square error further when $(\hat{f}_k-\bar{f}_k)^2 > \var(\bar{f}_k)$. For convenience, we use this approximate square error to approximate the variance, as shown in lines 3$\sim$4 in Algorithm~\ref{algo:numerical_error}.

For the weight vector $\mathbf{W}_k$ of node $n_k$, after \emph{weighted averaging}, it is updated as shown in lines 10$\sim$11 in Algorithm~\ref{algo:numerical_error}, with the optimized variance presented in Section~\ref{subsec:tree_refinement}. During \emph{frequency consistency}, node $n_k\in D_0$, which has a frequency close to or below zero, is set to zero, and its variance changes minimally. For node $n_k \in D_+$, its updated frequency is the weighted average of itself, its parent's frequency, and its sibling $n_s\in D_+$'s frequencies. As such, its weight can be updated as indicated in lines~14$\sim$15 in Algorithm~\ref{algo:numerical_error}.

\begin{algorithm}[htbp]
	\small
	\caption{Numerical Method for Variance Estimation}
	\label{algo:numerical_error}
	\KwIn{\textls[-1]{Frequencies $\hat{f}_k$ or $\bar{f}_k$ for each node and the node number $|\mathcal{T}|$}}
	\KwOut{The error variance of each node's refined frequency.}
	\tcp{Initialize $\mathbf{V}_{|\mathcal{T}| \times 1}$}
	Initialize vector $\mathbf{V}$\;
	\For{leaf node $n_k$}{
		$\var(\bar{f}_k) = \frac{4e^{\epsilon}}{\alpha_k \cdot N \cdot (e^{\epsilon} - 1)^2}$\;
		$\var(\hat{f}_k)\!\approx\!(\hat{f}_k\!-\!\bar{f}_k)^2\!>\!\var(\bar{f}_k) ? (\hat{f}_k\!-\!\bar{f}_k)^2 - \var(\bar{f}_k)\!:\!(\hat{f}_k\!-\!\bar{f}_k)^2$\;
		$v_k = \var(\dot{f}_k) = \frac{\var(\hat{f}_k) \var(\bar{f}_k)}{\var(\hat{f}_k) + \var(\bar{f}_k)}$\;
	}
	\For{non-leaf node $n_k$}{
		$v_k = \var(\bar{f}_k) = \frac{4e^{\epsilon}}{\alpha_k \cdot N \cdot (e^{\epsilon} - 1)^2}$\;
	}
	\tcp{Calculate $\mathbf{W}$}
	Initialize vector $\mathbf{W}_0 = \mathbf{0}_{|\mathcal{T}| \times 1}$ for the root and $\mathbf{W}_k$ with only the $k$-th value being one and others being zero for node $n_k$\;
	\For{non-leaf node $n_k$ visited in postorder traversal}{
		$\theta = \frac{\sum_{n_c \in \text{child}(n_k)} \var(\dot{f}_c)}{\var(\dot{f}_k) + \sum_{n_c \in \text{child}(n_k)} \var(\dot{f}_c)}$\;
		$\mathbf{W}_k = \theta \cdot \mathbf{W}_k + (1-\theta) \cdot (\sum_{n_c \in \text{child}(n_k)} \mathbf{W}_c)$\;
	}
	\For{node $n_k$ visited in level order traversal}{
		\If{$n_k \in D_+$}{
			Let $n_p$ represent $n_k$'s parent\;
			$\mathbf{W}_k = \frac{|D_+| - 1}{|D_+|} \cdot \mathbf{W}_k + \frac{1}{|D_+|} \cdot \mathbf{W}_p - \sum_{n_s \in D_+} \frac{1}{|D_+|} \cdot \mathbf{W}_s$\;
		}
	}
	\tcp{Calculate variances}
	\For{each node $n_k$}{
		$\var(\tilde{f}_k) = \sum_{1 \le j \le |\mathcal{T}|} w_{k,j}^2 \cdot v_j$\;
	}
	\KwRet $\var(\tilde{f}_k)$ of each node $n_k$ \;
\end{algorithm}

\section{Error Analysis on Multi-dimensional Queries}
\label{appendix:error_ana_md}

Multi-dimensional queries primarily depend on the estimated frequency of 2-D grids, which are adaptively constructed and refined based on the PriPL-Tree for each attribute. As such, we analyze the error in queries on these 2-D grids. As discussed in Section~\ref{subsec:adap_grid}, for a range query $Q$ that selects a portion $r$ of the area of a grid of size $g_i \times g_j$, the squared error is $2 rg_i g_j \sigma^2 \binom{m}{2} + r\eta \sum_{c \in G}f_c^2$, where $\sigma^2 = \frac{4e^{\epsilon}}{N \cdot (e^{\epsilon} - 1)^2}$. For simplicity, let $g$ represent the maximum grid partition size on each attribute; the asymptotic bound of the squared error is $O(\frac{r\cdot g^2 \cdot m^2}{N \cdot \epsilon^2})$. Furthermore, Theorem~\ref{theorem:md_error} provides the maximum absolute error, derived from Lemma~\ref{theorem:md_error_cell} and the analysis in Section~\ref{subsec:adap_grid}.
	
\begin{theorem}
\label{theorem:md_error}
	For any range query $Q$ selecting a portion $r$ of the area of a grid with size $g^2$, with at least $1-\beta$ probability, 
	\begin{equation}
	\max_{c}|\hat{Q}(r) - Q(r)| = O\left(\frac{rg^2m\sqrt{\log (g^2/\beta)}}{\epsilon \sqrt{N}}\right)
\end{equation}
\end{theorem}

\begin{lemma}
\label{theorem:md_error_cell}
For any cell $c$ in 2-D grids whose frequency is estimated by OUE with $\frac{N}{2 \cdot \binom{m}{2}}$ users, let $g$ denotes the maximum grid partition size on each dimension, $\hat{f}_c$ denote the estimated frequency and $f_c$ denotes the actual frequency, with at least $1-\beta$ probability, 
\begin{equation}
	\max_{c}|\hat{f}_c - f_c| = O\left(\frac{m\sqrt{\log (g^2/\beta)}}{\epsilon \sqrt{N}}\right)
\end{equation}
\end{lemma}

\begin{proof}
	For a frequency $\hat{f}_c$ estimated by OUE using $n = \frac{N}{2 \cdot \binom{m}{2}}$ users, it is unbiased and has a variance given by $\var(\hat{f}_c) = \frac{4e^{\epsilon}\cdot \binom{m}{2}}{N \cdot (e^{\epsilon} - 1)^2}$. 
	
	By Bernstein's inequality, when $\epsilon$ is small,
	\begin{align*}
		\Pr[|\hat{f}_c - f_c| \ge \delta]
		&\le 2 \cdot \exp\left( - \frac{n^2 \delta^2}{2\cdot n \cdot \frac{4 e^\epsilon}{(e^\epsilon - 1)^2} + \frac{1}{3} \cdot n \cdot \delta \cdot \frac{2\cdot e^\epsilon + 1}{e^\epsilon - 1} }\right) \\
		&= 2 \cdot \exp\left( - \frac{n \delta^2}{2\cdot O(\frac{1}{\epsilon^2}) + \frac{1}{3} \cdot \delta \cdot O(\frac{1}{\epsilon}) }\right).
	\end{align*}
	
	By the union bound, there exist $\delta\!=\!O\left(\frac{\sqrt{\log (g^2\!/\!\beta)}}{\epsilon\sqrt{n}}\right)\!=\!O\left(\frac{m\sqrt{\log (g^2\!/\!\beta)}}{\epsilon \sqrt{N}}\right)$ such that $\max_{c}|\hat{f}_c - f_c| \le \delta$ holds with at least $1-\beta$ probability.
\end{proof}

\section{Time Complexity Analysis}
\label{appendix:time_complexity}

In this section, we theoretically and experimentally compare the time complexity between PriPL-Tree and its competitors. For convenience, we assume all attributes have the same domain size $d$.

\subsection{1-D Time Complexity}
\label{appendix:time_complexity_1d}

\textbf{Theoretical Analysis of PriPL-Tree:} The construction time follows three phases. Phase 1 is relatively complex, involving the distribution estimation and private PL fitting two steps. The first step mainly includes the user values' aggregation time $O(N\cdot \alpha)$ and the distribution estimation time $O(d \cdot T)$ based on the EM or EMS algorithm within the SW mechansim, where $T$ denotes the number of iterations in EM and EMS. The second step includes matrix-based segment fitting, taking time $O(K^2 \cdot d)$, where calculating the inverse of matrices by the Gaussian elimination method, and interval partitioning, invoking at most $O(K\cdot d)$ segment fitting. When we apply the search acceleration strategy in interval partitioning with a granularity factor $\phi$ (refer to Section~\ref{subsubsec:interval_partitioning}), $O(K \cdot \phi \cdot \log_{\phi}d)$ times segment fitting is required. Next, in phase 2, involving the construction of the PriPL-Tree and the estimation of each node's frequency, the time complexity is $O(K + N\cdot (1-\alpha)\cdot K)$. Phase 3, refining the PriPL-Tree through two traversals, leads to a complexity of $O(K)$. Overall, the construction time complexity primarily stems from private user data aggregation, frequency estimation, and private PL fitting during phases 1 and 2, totaling $O(N\cdot K + d\cdot T + d\cdot\log d\cdot K^3)$. This has been confirmed across various datasets and domain sizes $d$, as illustrated in Figure~\ref{fig:runtime_1d_phase1}. The legend labeled ``LDP frequency estimation'' represents user data aggregation and frequency estimation using existing LDP mechanisms.

\begin{figure}[h!]
	\centering
	\includegraphics[width=0.45\textwidth]{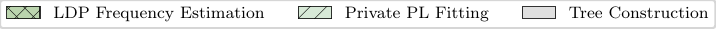}
	\vspace{-0.1in} 
	\\
	\subfloat[Varying Datasets with $d=1024$]{
	\includegraphics[width=0.23\textwidth,valign=t]{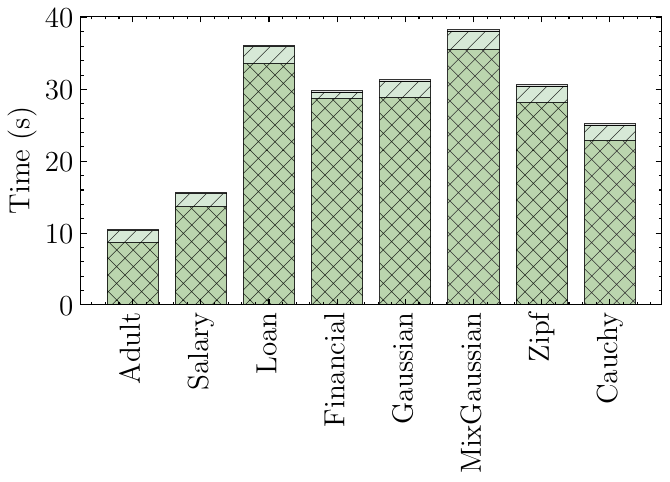}
	}
	\subfloat[Varying Domain Size $d$ on Gaussian]{
	\includegraphics[width=0.23\textwidth,valign=t]{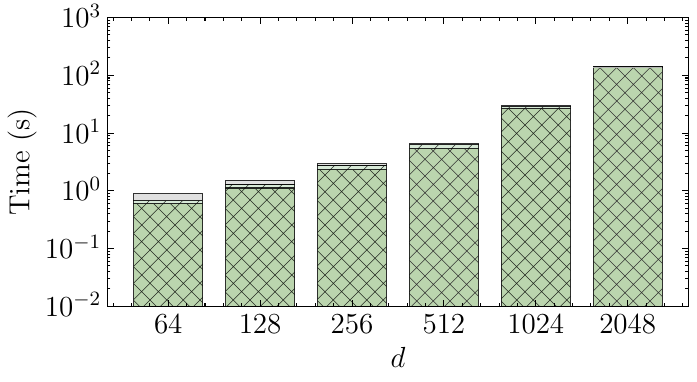}
	}
	\caption{Construction Time of PriPL-Tree}
	\label{fig:runtime_1d_phase1}
\end{figure}

\textbf{Theoretical Analysis of Competitors:} In Table~\ref{tab:performance_1d}, we compare the construction and query time complexities of different methods in 1-D scenarios. Because all compared methods utilize hierarchical tree structures, the construction time complexity of them primarily stems from processing users' reports, while the query time complexity depends on the tree height. In this table, the construction time complexities for AHEAD and DHT are sourced from \cite{du2021ahead}, and that for PrivNUD is analyzed by us using the same methodology. Since that the maximum segment number $K$ in PriPL-Tree is typically small, our construction time complexity is competitive with others, and our query time is significantly lower than our competitors.

\begin{table}[h!]
\small 
\caption{Time Complexity Comparison for 1-D Range Queries}
\centering
\begin{tabular}{|c|c|c|}
\hline
\textbf{Methods} & \textbf{Construction Time} & \textbf{Query Time}\\ 
\hline
PriPL-Tree & $O(N\cdot K + d\cdot T + d\cdot \log d\cdot K^3)$ & $O(\log_2 K)$ \\
\hline
PrivNUD & $O(\log_d 2 \cdot N \cdot d)$  & $O(\log_2 d)$\\
\hline
AHEAD & $O(\log_d 2 \cdot N \cdot d)$ & $O(\log_2 d)$\\
\hline
DHT & $O(N + d^3)$ &  $O(\log_2 d)$ \\
\hline
\end{tabular}
\label{tab:performance_1d}
\end{table}

\textbf{Experimental Comparison:} To ensure a fair comparison of method runtimes, we reimplemented DHT, originally in C++, in Python to match the programming language of the other methods. We have analyzed the construction and query times across different datasets in Section~\ref{subsec:1d_exp}. Additionally, we evaluate the construction times of different methods on varying domain sizes $d$ and the query times on varying query volumes.

For construction time, all methods are mainly influenced by the time of aggregating users' perturbed messages, a step inherent in employed LDP frequency estimation mechanisms, and thus increase with domain size $d$. In our comparisons, AHEAD \cite{du2021ahead} and PrivNUD \cite{wang2023privnud} share the same asymptotic time complexity in Table~\ref{tab:performance_1d}, but show significant discrepancies in experimental performance. This variation stems from their different implementations, where AHEAD utilizes the Treelib package for tree structures, similar to ours, while PrivNUD uses a nested array.

For query time, DHT shows a rising trend with increasing query volume $Vol(Q)$, whereas others, including PriPL-Tree, AHEAD, and PrivNUD, initially rise then decline as query volume increases. This trend is significantly pronounced for PrivNUD but slight for PriPL-Tree and AHEAD. This observation aligns with the theoretical analysis that in the hierarchical trees, the query time correlates with the differing numbers of nodes accessed at various query volumes, typically peaking when the query volume is around $50\%$.  

\begin{figure}[h]
	\centering
	\includegraphics[width=0.35\textwidth]{figures/experiments/legend_1d_runtime.pdf}
	\vspace{-0.1in} 
	\\
	\subfloat[Varying Domain Size $d$]{
	\includegraphics[width=0.23\textwidth,valign=t]{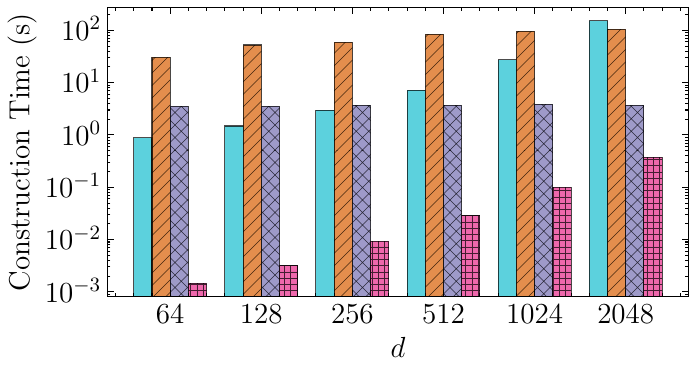}
	}
	\subfloat[Varying Query Volume $vol(Q)$]{
	\includegraphics[width=0.23\textwidth,valign=t]{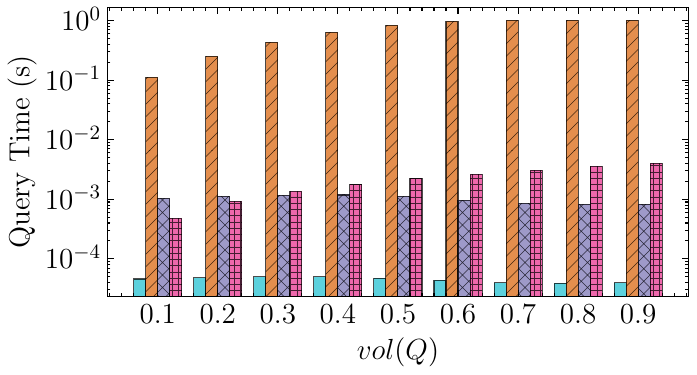}
	}
	\caption{Runtime Evaluation on Gaussian}
	\label{fig:runtime_1d_varying_para}
\end{figure}

\subsection{Multi-D Time Complexity}
\label{appendix:time_complexity_md}

\textbf{Theoretical Analysis:} In multi-dimensional scenarios, leveraging the 1-D and 2-D estimations is a common strategy across all methods. The construction time complexity of these hybrid low-dimensional data structures is primarily from processing $N$ user's reports, where each user reports multiple perturbed counts for nodes in the tree or cells in the grid. The query time complexity arises from computing the frequencies of the corresponding $\binom{\lambda}{2}$ 2-D grids and generating the response matrix with $2^\lambda$ values. Because generating the response matrix with given $\binom{\lambda}{2}$ 2-D grids is uniform for all methods, we focus on the time complexity of 2-D grids estimation for comparison, omitting the time consumption of the response matrix step. Following this idea, we analyze the time complexity of our PriPL-Tree and other competitors in Table~\ref{tab:performance_md}. During our analysis, we note that the granularity of the grids in the PriPL-Tree method is typically smaller than the maximum number of segments per attribute, denoted as $K$; therefore, we set the granularity to $O(K)$. For other methods, we use the optimized granularity reported in their respective papers.

\begin{table}[h]
\small 
\setlength\tabcolsep{1pt} 
\caption{Time Complexity Comparison for $\lambda$-d Range Queries on $m$-d Datasets}
\centering
\begin{tabular}{|c|c|c|}
\hline
\textbf{Methods} & \textbf{Construction Time} & \textbf{Query Time}\\ 
\hline
PriPL-Tree & $O(N\cdot K^2 + m\cdot d\cdot T)$ & $O(\lambda^2\cdot K^2)$ \\
\hline
PrivNUD & $O(\log_d 2 \cdot N \cdot d / m + N\sqrt{N}/m)$ & $O(\lambda^2\cdot \sqrt{N} / m)$ \\
\hline
AHEAD & $O(\log_d 4\cdot N \cdot d^2)$ & $O(\lambda^2 \cdot\log_4 d^2)$ \\
\hline
HDG & $O(N\sqrt{N}/m)$ & $O(\lambda^2\cdot \sqrt{N} / m)$ \\
\hline
PRISM & $O(N\sqrt{N}/m)$ & $O(\lambda^2\cdot \sqrt{N} / m)$ \\
\hline
\multicolumn{3}{l}{\makecell[l]{$*$: $K$ is the maximum segment number in PriPL-Tree, and $T$ is the number\\ of iterations in the EM algorithm within the SM protocol \cite{li2020estimating}. }}\\
\end{tabular}
\label{tab:performance_md}
\end{table}

\textbf{Experimental Comparison:} In Figure~\ref{fig:runtime_md}, we assess the actual runtime of these methods. For construction time across eight different datasets, i.e., Figure~\ref{fig:runtime_md} (a), our PriPL-Tree performs comparably to PrivNUD and is slightly longer than HDG, which is the fastest and only uses grid structures without trees. In Figure~\ref{fig:runtime_md} (c), the construction times for all methods increase with data dimensionality. PRISM is an exception, with significantly larger time consumption in 2-D cases. This is because it generates significantly finer grids in the 2-D case than in other dimensions with its granularity optimization strategy \cite{wang2022prism}. For query time, i.e., Figure~\ref{fig:runtime_md} (b) and (d), our PriPL-Tree method consistently achieves the fastest responses, with times increasing as query dimension $\lambda$ increases.

\begin{figure}[ht]
	\centering
	\includegraphics[width=0.45\textwidth]{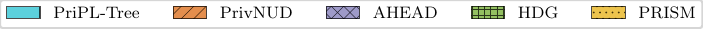}
	\vspace{-0.1in} 
	\\
	\subfloat[Construction Time]{
	\includegraphics[width=0.23\textwidth,valign=t]{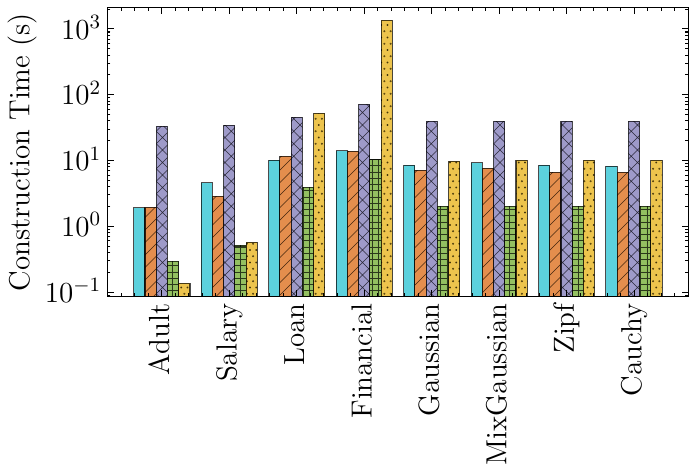}
	}
	\subfloat[Query Time]{
	\includegraphics[width=0.23\textwidth,valign=t]{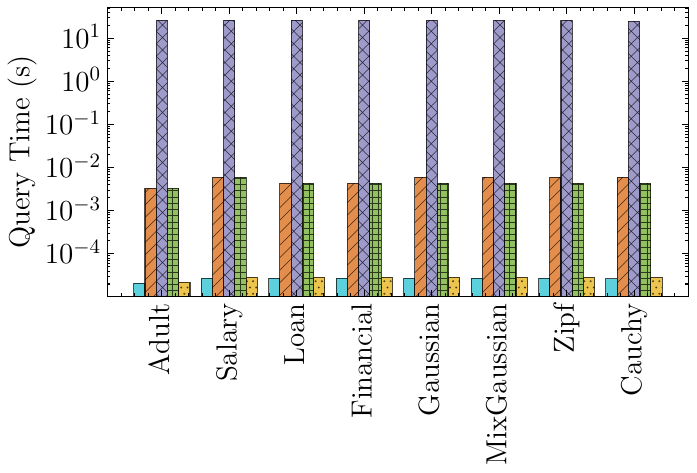}
	}
	\\
	\subfloat[Construction Time on Gaussian]{
	\includegraphics[width=0.23\textwidth,valign=t]{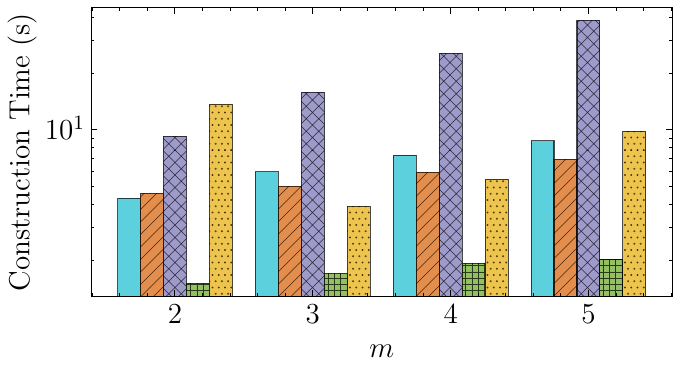}
	}
	\subfloat[Query Time on Gaussian]{
	\includegraphics[width=0.23\textwidth,valign=t]{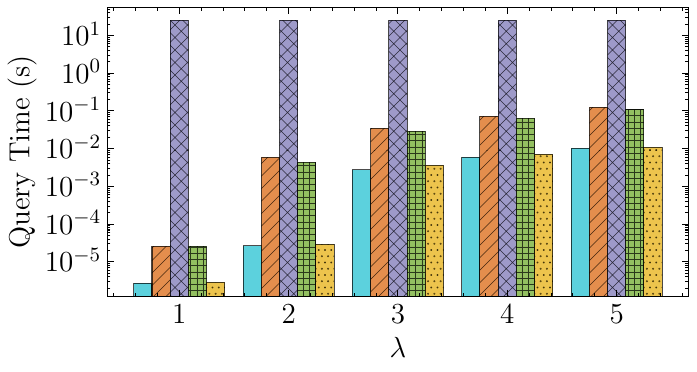}
	}
	\caption{Runtime Evaluations (Defaulty, $m = 5$ and $\lambda = 2$)}
	\label{fig:runtime_md}
\end{figure}

\section{Datasets Description}
\label{appendix:datasets}

In this section, we provide additional descriptions of both synthetic and real-world datasets. For our main 1-D scenario, we illustrate the distribution for the default attribute of each dataset in Figure~\ref{fig:datasets}. For multi-dimensional scenarios, we present the mean, variance, skewness, and kurtosis for each attribute across all datasets. Skewness indicates data distribution asymmetry, with positive values suggesting a long right tail and negative values a long left tail. Kurtosis measures the peak sharpness and tail thickness of the distribution. We bold the Kurtosis values exceeding 3 in the following tables, indicating an excessively peaked distribution. Given the similar statistical characteristics across dimensions, we summarize the average indicators for synthetic datasets in Table~\ref{tab:synthetic}. For real-world datasets, we detail these characteristics per attribute in the following: Adult dataset in Table~\ref{tab:adult}, Loan dataset in Table~\ref{tab:loan}, Salary dataset in Table~\ref{tab:salary}, and Financial dataset in Table~\ref{tab:financial}.

\begin{figure}
	\centering
	\subfloat[Gaussian]{
		\includegraphics[width=0.24\textwidth]{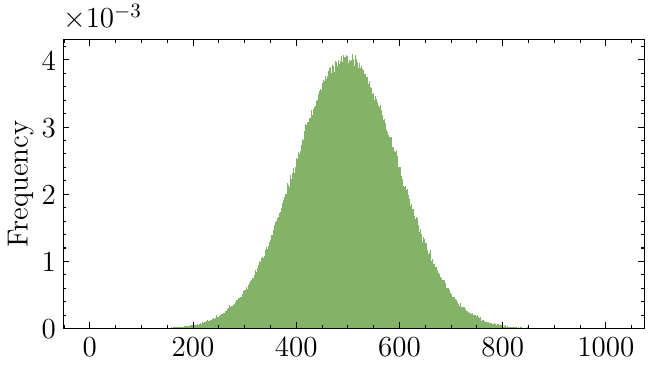}
	}
	\subfloat[MixGaussian]{
		\includegraphics[width=0.24\textwidth]{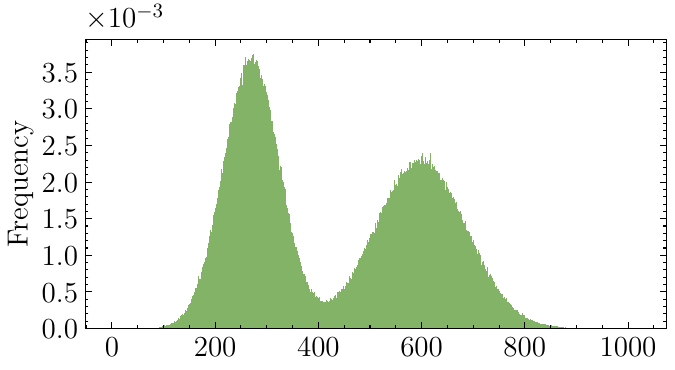}
	}
	\\
	\subfloat[Cauchy]{
		\includegraphics[width=0.24\textwidth]{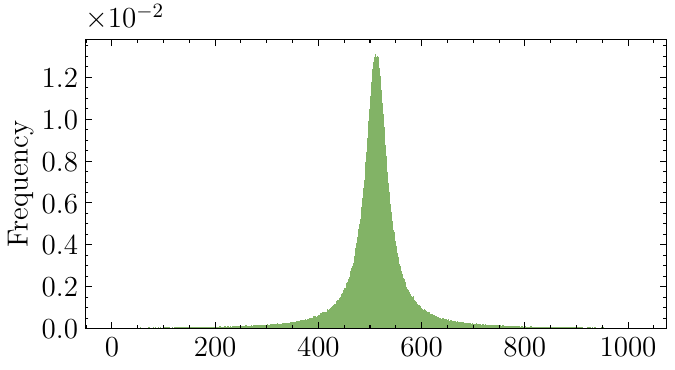}
	}
	\subfloat[Zipf]{
		\includegraphics[width=0.24\textwidth]{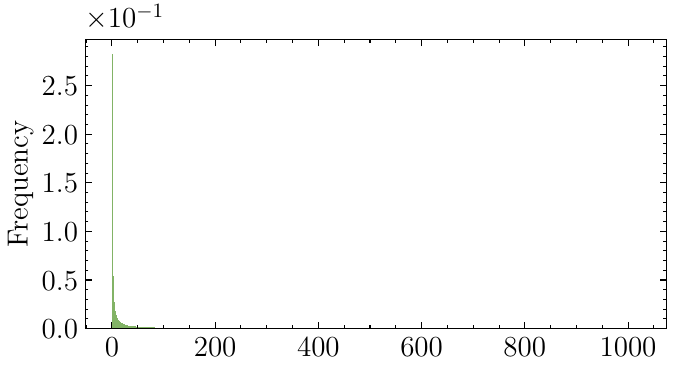}
	}
	\\
	\subfloat[Adult (``fnlwgt'')]{
		\includegraphics[width=0.24\textwidth]{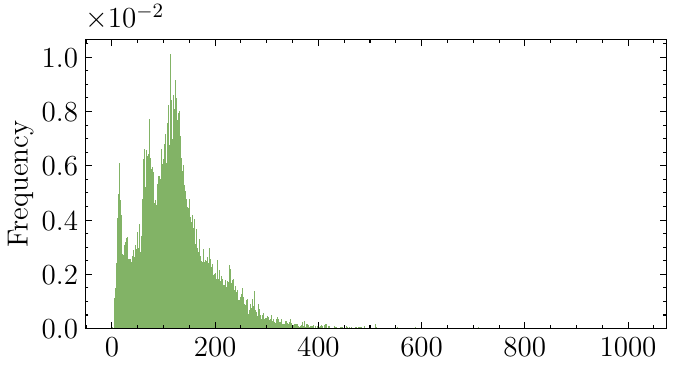}
	}
	\subfloat[Loan (``total\_pymnt'')]{
		\includegraphics[width=0.24\textwidth]{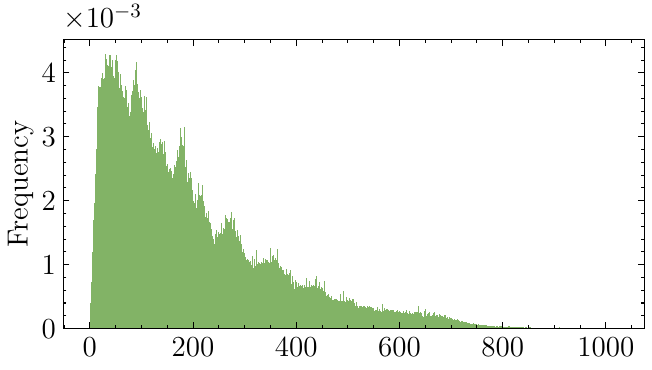}
	}
	\\
	\subfloat[Salary (``TotalPay'')]{
		\includegraphics[width=0.24\textwidth]{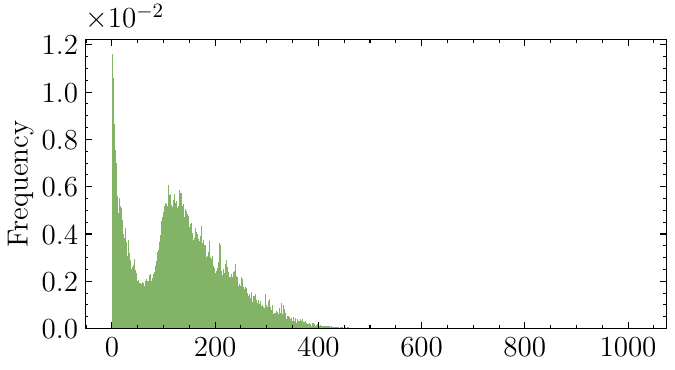}
	}
	\subfloat[Financial (``amount'' )]{
		\includegraphics[width=0.24\textwidth]{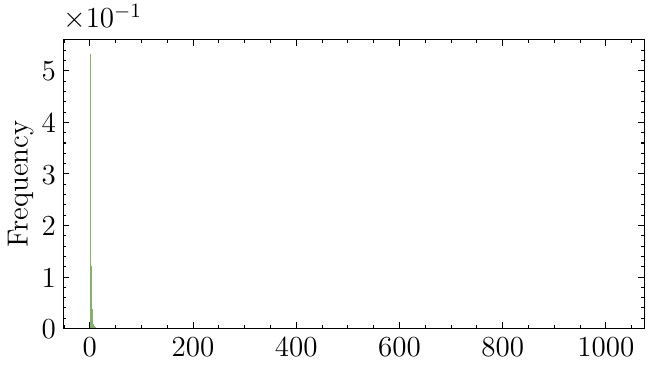}
	}
	\caption{Data Distributions of 1-D Datasets}
	\label{fig:datasets}
\end{figure}

\begin{table}[h]
\caption{Summary of Synthetic Datasets ($d=256$)}
\centering
\begin{tabular}{ccccc}
\hline
\textbf{Dataset} & \textbf{Mean} & \textbf{Variance} & \textbf{Skewness} & \textbf{Kurtosis} \\
\hline
Gaussian       & 155.0        & 775.34      & 1.44           & -0.28          \\
MixGaussian     & 135.12        & 2516.62      & 1.15           & -0.28          \\
Cauchy     & 159.37        & 609.50      & 4.42           & \textbf{15.46}          \\
Zipf     & 24.40       & 2517.47      & 18.58           & \textbf{284.47}          \\
\hline
\end{tabular}
\label{tab:synthetic}
\end{table}

\begin{table}[h]
\caption{Summary of Adult Dataset ($d=256$)}
\centering
\begin{tabular}{ccccc}
\hline
\textbf{Attribute} & \textbf{Mean} & \textbf{Variance} & \textbf{Skewness} & \textbf{Kurtosis} \\
\hline
fnlwgt           & 30.36         & 336.88                     & 2.19   & \textbf{4.03}            \\
age          & 21.58        & 186.06                      & -0.03  & -1.64           \\
capital-gain            & 2.71        & 353.74                      & 15.90 & \textbf{250.75}         \\
capital-loss            & 5.11       & 555.48                      & 15.90  & \textbf{250.93}          \\
hours-per-week            & 39.44       & 152.45                    & 8.82   & \textbf{80.72}    
\\
\hline      
\end{tabular}
\label{tab:adult}
\end{table}

\begin{table}[h]
\caption{Summary of Loan Dataset ($d=256$)}
\centering
\begin{tabular}{ccccc}
\hline
\textbf{Attribute}             & \textbf{Mean} & \textbf{Variance} & \textbf{Skewness} & \textbf{Kurtosis} \\
\hline
total\_pymnt                  & 48.67      &1620.19    & 1.27           & 0.35           \\
total\_rec\_int                & 21.69     &596.95     & 2.98           & \textbf{8.38}            \\
installment                   & 66.01     &1595.08     & 1.12           & 0.37            \\
\makecell[l]{total\_il\_high\\ \_credit\_limit} & 5.02     & 29.44     & 5.61             & \textbf{32.28}           \\
loan\_amnt                    & 93.03     &3661.59     & 4.75           & \textbf{26.27}          
\\
\hline
\end{tabular}
\label{tab:loan}
\end{table}

\begin{table}[h]
\caption{Summary of Salary Dataset ($d=256$)}
\centering
\begin{tabular}{ccccc}
\hline
\textbf{Attribute} & \textbf{Mean} & \textbf{Variance} & \textbf{Skewness} & \textbf{Kurtosis} \\
\hline
TotalPay           & 33.59      &515.88     & 1.80            & 2.56            \\
TotalPayBenefits   & 42.15      &797.60     & 1.64            & 2.64             \\
BasePay            & 52.80      &1174.51       & 1.46            & 2.33            \\
OvertimePay        & 5.08      &140.77       & 15.43           & \textbf{240.57}          \\
OtherPay           & 6.22      &25.47      & 13.47           & \textbf{194.04}           
\\
\hline
\end{tabular}
\label{tab:salary}
\end{table}

\begin{table}[h]
\caption{Summary of Financial Dataset ($d=256$)}
\centering
\begin{tabular}{ccccc}
\hline
\textbf{Attribute} & \textbf{Mean} & \textbf{Variance} & \textbf{Skewness} & \textbf{Kurtosis} \\
\hline
amount             & 0.23      &2.65      & 15.76          & \textbf{247.68}       \\
oldbalanceOrg    & 3.40      & 151.81      & 15.82         & \textbf{249.22}       \\
newbalanceOrig     & 4.24      & 225.10      & 15.72           & \textbf{246.97}          \\
oldbalanceDest     & 0.57      & 5.61     & 15.57          & \textbf{243.49}         \\
newbalanceDest   & 0.64      & 6.60      & 15.47         & \textbf{241.30}         
\\
\hline
\end{tabular}
\label{tab:financial}
\end{table}

\section{Additional Evaluation on High-dimensional Dataset}

In this section, we evaluate performance on the high-dimensional IPUMS dataset, as shown in Figure~\ref{fig:high_dimension_evaluation}. This dataset, sourced from the 2022 IPUMS repository \footnote{https://usa.ipums.org/usa/}, contains 50 dimensions and 15,721,123 samples. Due to the high runtime complexity of AHEAD, as detailed in Appendix~\ref{appendix:time_complexity_md}, it could not produce effective results on 50-D datasets within the limited time in our experimental setup, so we excluded its results from the figure. This omission does not impact our analysis, as previous experiments have already demonstrated its inferior performance compared to PrivNUD and our PriPL-Tree.

Overall, PriPL-Tree consistently achieves the lowest MSE among all competitors. For various privacy budgets (Figure~\ref{fig:high_dimension_evaluation} (a)), PriPL-Tree shows a reduction in MSE ranging from 46.34\% to 72.22\%, with an average reduction of 60.67\%. Across different query dimensions (Figure~\ref{fig:high_dimension_evaluation} (b)), it reduces MSE by 56.80\% to 91.72\%, averaging 69.90\%. The superiority of PriPL-Tree here is more pronounced than in experiments on only 5-D datasets, as shown in Figures~\ref{fig:evaluation_md_epsilon} and \ref{fig:evaluation_md_params}.

Furthermore, Figure~\ref{fig:high_dimension_evaluation} (b) shows a decrease in MSE with increasing query dimension $\lambda$, in contrast to the results from 5-D synthetic Gaussian datasets depicted in Figure~\ref{fig:evaluation_md_params} (b), where MSE increases with $\lambda$. This discrepancy arises from the different characteristics of the datasets. Methods estimating high-dimensional range queries using low-dimensional responses are more accurate when attribute correlations are minor. The Gaussian dataset, with high inter-attribute correlation (covariance of 0.6), shows decreasing accuracy as $\lambda$ increases. In contrast, IPUMS, with generally lower correlations among attributes, allows for relatively precise estimates across varying dimensions. The observed decrease in MSE is primarily due to the lower actual frequency of higher-dimensional queries, where the error is proportional to this frequency.
\begin{figure}[t]
	\centering
	\includegraphics[width=0.45\textwidth]{figures/experiments/legend_md.pdf}
	\\
	\subfloat[Varying Privacy Budget $\epsilon$]{
	\includegraphics[width=0.23\textwidth,valign=t]{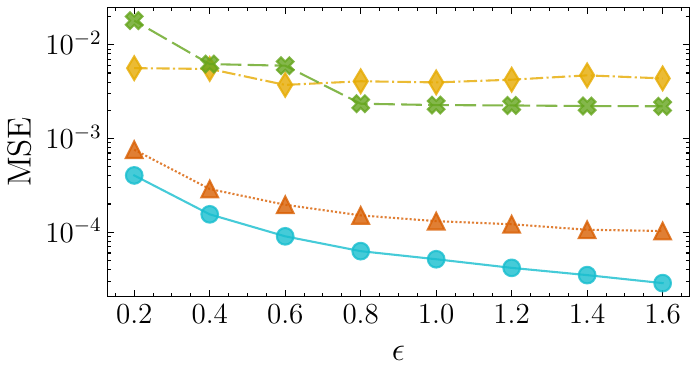}
	}
	\subfloat[Varying Query Dimension $\lambda$]{
	\includegraphics[width=0.23\textwidth,valign=t]{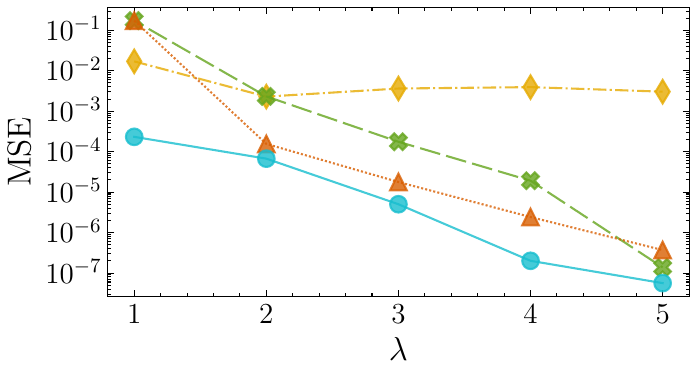}
	}
	\caption{Evaluation of Range Queries on 50-D IPUMS}
	\label{fig:high_dimension_evaluation}
\end{figure}

\end{document}